%% file: paper-lmcs.tex
\documentclass{lmcs} %%% last changed 2014-08-20


\usepackage{hyperref}
\theoremstyle{plain} %\autorefname{satz}{Satz}{S\"atze}
\usepackage{amssymb}
\usepackage{xspace}
\usepackage{etoolbox}
\usepackage{mathtools}
\usepackage{scalerel}
\usepackage{stmaryrd}
\usepackage{hyperref}
\usepackage{diagbox}
\usepackage{amsmath}

\newcommand{\ALC}{\ensuremath{\mathcal{ALC}}\xspace}
\newcommand{\FLzero}{\ensuremath{\mathcal{FL}_0}\xspace}
\newcommand{\EL}{\ensuremath{\mathcal{E\mkern-1.618mu L}}\xspace}

\newcommand{\id}{\approx}
\newcommand{\idEL}{\approx_\EL}
\newcommand{\idFLzero}{\approx_{\FLzero}}
\newcommand{\leqEL}{\leq_\EL^V}
\newcommand{\simEL}{\sim_\EL^V}

\DeclareMathOperator{\Dom}{Dom}
\DeclareMathOperator{\Sig}{Sig}
\DeclareMathOperator{\VRan}{VRan}
\DeclareMathOperator{\Var}{Var}
\DeclareMathOperator{\rd}{rd}
\DeclareMathOperator{\VR}{VRan}

\makeatletter

\def\math@Alphabet{A,B,C,D,E,F,G,H,I,J,K,L,M,N,O,P,Q,R,S,T,U,V,W,X,Y,Z}
\def\math@alphabet{a,b,c,d,e,f,g,h,i,j,k,l,m,n,o,p,q,r,s,t,u,v,w,x,y,z}
\newcommand{\math@def}[2]{\expandafter#1\expandafter{\csname#2\endcsname}}
\newcommand{\math@forcsvlist}[2]{\expandafter\forcsvlist\expandafter#1\expandafter{#2}}

\newcommand{\math@defmc}[1]{\math@def{\newcommand}{#1mc}{\ensuremath{\mathcal{#1}}\xspace}}
\newcommand{\math@defmf}[1]{\math@def{\newcommand}{#1mf}{\ensuremath{\mathfrak{#1}}\xspace}}
\newcommand{\math@defsf}[1]{\math@def{\newcommand}{#1sf}{\ensuremath{\mathsf{#1}}\xspace}}
\newcommand{\math@defbb}[1]{\math@def{\newcommand}{#1bb}{\ensuremath{\mathbb{#1}}\xspace}}
\newcommand{\math@defbf}[1]{\math@def{\newcommand}{#1bf}{\ensuremath{\mathbf{#1}}\xspace}}
\newcommand{\math@deful}[1]{\math@def{\newcommand}{#1ul}{\ensuremath{\underline{#1}}\xspace}}

\math@forcsvlist{\math@defmc}{\math@Alphabet}
\math@forcsvlist{\math@defmf}{\math@Alphabet}
\math@forcsvlist{\math@defsf}{\math@Alphabet}
\math@forcsvlist{\math@defbb}{\math@Alphabet}
\math@forcsvlist{\math@defbf}{\math@Alphabet}
\math@forcsvlist{\math@deful}{\math@Alphabet}
\math@forcsvlist{\math@defmc}{\math@alphabet}
\math@forcsvlist{\math@defmf}{\math@alphabet}
\math@forcsvlist{\math@defbb}{\math@alphabet}
\math@forcsvlist{\math@defbf}{\math@alphabet}
\math@forcsvlist{\math@deful}{\math@alphabet}

\newcommand{\ModK}{\ensuremath{\mathsf{K}}\xspace}
\newcommand{\Idem}{\ensuremath{\mathsf{I}}\xspace}
\newcommand{\Assoc}{\ensuremath{\mathsf{A}}\xspace}
\newcommand{\Unit}{\ensuremath{\mathsf{U}}\xspace}
\newcommand{\Commu}{\ensuremath{\mathsf{C}}\xspace}
\newcommand{\AC}{\ensuremath{\mathsf{AC}}\xspace}
\newcommand{\AU}{\ensuremath{\mathsf{AU}}\xspace}
\newcommand{\ACU}{\ensuremath{\mathsf{ACU}}\xspace}
\newcommand{\ACI}{\ensuremath{\mathsf{ACI}}\xspace}
\newcommand{\AI}{\ensuremath{\mathsf{AI}}\xspace}
\newcommand{\D}{\ensuremath{\mathsf{D}}\xspace}
\newcommand{\ACUI}{\ensuremath{\mathsf{ACUI}}\xspace}
\newcommand{\bSLmO}{\ensuremath{\mathsf{bSLmO}}\xspace}
\newcommand{\ACUIh}{\ensuremath{\mathsf{ACUIh}}\xspace}
\newcommand{\ACUh}{\ensuremath{\mathsf{ACUh}}\xspace}

\newcommand{\AG}{\ensuremath{\mathsf{AG}}\xspace}
\newcommand{\AGh}{\ensuremath{\mathsf{AGh}}\xspace}
\newcommand{\GAUSS}{\ensuremath{\mathsf{GAUSS}}\xspace}
\newcommand{\Boole}{\ensuremath{\mathsf{B}}\xspace}

\newcommand{\unit}{0}
\newcommand{\seq}{\mathfrak{s}}
\newcommand{\allPairs}[1]{P_{#1}}
\newcommand{\grRel}[1]{>_{#1}}
\newcommand{\occVar}[2]{\sV^{#1}_{#2}}
\newcommand{\inv}{\mathsf{inv}}

\newcommand{\zero}{0}
\newcommand{\one}{1}
\newcommand{\bRingSum}{+}
\newcommand{\bRingMult}{*}

\newcommand{\sR}{\Smc}
\newcommand{\BS}{\Smc_{\ACUI}}
\newcommand{\sV}{a}
\newcommand{\sVb}{b}
\newcommand{\sVc}{c}
\newcommand{\sPreSymbol}{S}
\newcommand{\sPre}[1]{\leq_{#1}}
\newcommand{\sPreGr}[1]{>_{#1}}
\newcommand{\sPreL}[1]{<_{#1}}
\newcommand{\sRPr}{\otimes}
\newcommand{\sRSum}{\oplus}
\newcommand{\sRBSum}{\bigoplus}
\newcommand{\sZero}{\mathsf{0}}
\newcommand{\sOne}{\mathsf{1}}
\newcommand{\conc}{\cdot}
\newcommand{\eword}{\varepsilon}
\newcommand{\fSets}{\Wmc}
\newcommand{\alphb}{\Delta}
\newcommand{\semNPol}{\mathbb{N}[X]}
\newcommand{\NPol}{\mathbb{N}[X]}
\DeclareMathOperator{\multPol}{\overline{\cdot}}
\DeclareMathOperator{\addPol}{\overline{+}}

\newcommand{\spEName}{restrictive}

\begin{document}

% If the title is longer than 55 characters, then specify a shorter running title as the optional argument to \title. The running title should be roughyl at most 55 characters:
\title[Unification Types and Instantiation Preorders]{The Unification Type of an Equational Theory May Depend on the Instantiation Preorder: %\\
From Results for Single Theories to Results for Classes of Theories {\lsuper*}}
\titlecomment{{\lsuper*}This article is an extended version of a paper published at the conference FSCD 2025 \cite{BaaderG25}.}
\thanks{Both authors were partially supported by DFG grant 389792660 as part of TRR 248 -- CPEC, and by the German Federal Ministry of Education and Research (BMBF, SCADS22B) and the Saxon State Ministry for Science, Culture and Tourism (SMWK) by
funding the competence center for Big Data and AI “ScaDS.AI Dresden/Leipzig”.}	%optional

% affiliations are numbered automatically with a, b, c (see below)
% use the optional argument to indicate the affiliation(s) of each author
% omit the argument if there is only one author, or only one affiliation
\author[F.~Baader]{Franz Baader\lmcsorcid{0000-0002-4049-221X}}[a,b]
\author[O.~Fern\'andez Gil]{Oliver {Fern\'andez Gil\lmcsorcid{0000-0002-9458-1701}}}[a]

% affiliation 1 (automatically numbered a)
\address{Theoretical Computer Science, TU Dresden, Germany}	%optional
% write emails for all authors having that affiliation
\email{franz.baader@tu-dresden.de, oliver.fernandez@tu-dresden.de}  %optional

% affiliation 2 (automatically numbered b)
\address{Center for Scalable Data Analytics and Artificial Intelligence (ScaDS.AI) Dresden/Leipzig, Germany}	%optional
%\email{name2@email2}  %optional

%% etc.

%% required for running head on odd and even pages, use suitable
%% abbreviations in case of long titles and many authors:

%%%%%%%%%%%%%%%%%%%%%%%%%%%%%%%%%%%%%%%%%%%%%%%%%%%%%%%%%%%%%%%%%%%%%%%%%%%

%% the abstract has to PRECEDE the command \maketitle:
%% be sure not to issue the \maketitle command twice!

\begin{abstract}
  The unification type of an equational theory is defined using a preorder on substitutions, called the instantiation preorder, whose scope is either restricted to the variables occurring in the unification problem, or unrestricted such that all variables are considered. It has been known for more than three decades that the unification type of an equational theory may vary, depending on which instantiation preorder is used. More precisely, it was shown in 1991 that the theory \ACUI of an associative, commutative, and idempotent binary function symbol with a unit is unitary w.r.t.\ the restricted instantiation preorder, but not unitary w.r.t.\ the unrestricted one. In 2016 this result was strengthened by showing that the unrestricted type of this theory also cannot be finitary.  In the conference version of this article, we considerably improved on this result by proving that \ACUI is infinitary w.r.t.\ the unrestricted instantiation preorder, thus precluding type zero. We also showed that, w.r.t.\ this preorder, the unification type of \ACU (where idempotency is removed from the axioms) and of \AC (where additionally the unit is removed) is infinitary, though it is respectively unitary and finitary in the restricted case. In the other direction, we proved (using the example of unification in the description logic \EL) that the unification type may actually improve from type zero to infinitary when switching from the restricted instantiation preorder to the unrestricted one. 
In the present article, we not only determine the unrestricted unification type of considerably more equational theories, but we also prove general
results for whole classes of theories. In particular, we show that theories that are regular and finite, regular and locally finite, or regular, monoidal, and
satisfy an additional condition are Noetherian, and thus cannot have unrestricted unification type zero.
\end{abstract}

\maketitle

\input{main}

\section*{Acknowledgment}
\noindent The authors would like to thank the reviewers (of the original FSCD 2025 paper \cite{BaaderG25} and of this journal version) for their valuable comments.

\bibliographystyle{alphaurl}
\bibliography{literature}

%\end{thebibliography}

%\appendix

\end{document}

%% file: main.tex
\section{Introduction}

Syntactic unification of terms was independently introduced by Robinson~\cite{Robinson65} and Knuth and Bendix~\cite{KnuthBendix70}
as a tool for computing resolvents in resolution-based theorem proving and critical pairs in the completion of term rewriting systems.
Both showed the important result that any solvable unification problem has a most general unifier (mgu), i.e., a unifier that has all other unifiers as 
instances. In these papers, a substitution $\theta$ is defined to be an instance of a substitution $\sigma$ if there is a substitution $\lambda$
such that $\lambda\sigma = \theta$, i.e., $\lambda(\sigma(x)) = \theta(x)$ holds for all variables $x$ in the countably infinite set of
variables $V$ available for building terms. In this paper, we call the preorder on substitutions obtained this way the unrestricted
instantiation preorder and write it as $\sigma \le_\emptyset^V\theta$, where the index $\emptyset$ indicates that terms are to be made syntactically equal.
In addition to many papers on how to compute the mgu efficiently 
(e.g.,~\cite{DBLP:journals/jcss/PatersonW78,DBLP:journals/toplas/MartelliM82,BidoitCorbin83}),
properties of the preorder on substitutions defined this way have, for instance, been investigated
in~\cite{DBLP:journals/jsc/Eder85,DBLP:books/mk/minker88/LassezMM88}.
 
In his seminal paper~\cite{Plotkin72}, Plotkin proposed to build certain equational theories (such as associativity or commutativity) into
the unification algorithm rather than treating their axiomatization within the general theorem proving process. Similar proposals were also
made in the setting of Knuth-Bendix completion in term rewriting~\cite{Peterson:Stickel81,Jouannaud:Kirchner86}. As already pointed out
by Plotkin, in the equational setting most general unifiers need not exist and their r\^ole is instead taken on by minimal complete sets of unifiers,
i.e., sets of unifiers such that every unifier is an instance of a unifier in this set, and no distinct elements in the set are comparable w.r.t.\
the instantiation preorder.\footnote{%
Plotkin actually calls these sets ``maximally general set of unifiers'' and requires two additional technical conditions.
}
As instantiation preorder he uses what we call the restricted instantiation preorder, i.e., $\sigma \le_E^X\theta$ where $E$ is the theory
modulo which unification is considered and $X$ is the set of variables occurring in the unification problem. This preorder requires the existence
of a substitution $\lambda$ such that $\lambda(\sigma(x)) \id_E \theta(x)$ holds for all variables $x\in X$. He explains the
use of equality modulo $E$ ($\id_E$) in this definition, but does not comment on the restriction to the variables of the unification problem.
Plotkin also gives an example of an equational theory (associativity \Assoc) where minimal complete sets of unifiers may become infinite, and conjectures
that there may exist theories for which such sets do not exist.

Siekmann and colleagues~\cite{DBLP:journals/cca/RaulefsSSU79} proposed to characterize equational theories according to the cardinality and existence 
of minimal complete sets of unifiers into the types unitary, finitary, infinitary, and zero. However, in the first overview paper on
results in this direction~\cite{DBLP:journals/cca/RaulefsSSU79}, the unrestricted instantiation preorder is used, and the same is true for
Siekmann's work on unification modulo commutativity~\cite{DBLP:conf/eurosam/Siekmann79a}. In later overview papers~\cite{DBLP:conf/cade/Siekmann84,DBLP:conf/ecai/Siekmann86,DBLP:journals/jsc/Siekmann89} he describes the unrestricted
instantiation preorder in the introduction, but employs the restricted one in the formal definition of unification types, again without
explanation.
Due to potential applications of equational unification in resolution-based theorem proving and term rewriting, unification properties (among them
the unification type) of frequently encountered equational axioms such as associativity, commutative, idempotency, distributivity and their 
combinations were extensively studied in the automated deduction community in the 1980s and 1990s 
(see \cite{JoKi91,DBLP:books/ox/LAI2Hb1994/BaaderS94,DBLP:books/el/RV01/BaaderS01} for overviews). 
More recently, unification in certain logics such as modal and description logics has drawn considerable 
interest~\cite{DBLP:journals/jsyml/Ghilardi99,DBLP:journals/corr/abs-1006-2289,DBLP:journals/igpl/BaaderG11}, where the goal is to
make a formula valid or two formulas equivalent by applying a substitution. In particular, the unification types of various modal logics
have been determined (see, e.g.,~\cite{DBLP:journals/logcom/Jerabek15,DBLP:journals/ndjfl/Iemhoff16,DBLP:journals/apal/DzikKW22,DBLP:journals/igpl/BalbianiGRT23,DBLP:journals/igpl/AlizadehABM23,Balbiani2024,DBLP:journals/apal/DzikKW25}).
In both areas, the authors usually employ the restricted instantiation preorder, though~\cite{DBLP:journals/logcom/Jerabek15} is an exception.

In the conference version~\cite{BaaderG25} of the present paper, we started to investigate the impact that using the unrestricted rather than the restricted instantiation preorder has on the unification
type. Until then, there were only two partial results in this direction. Already in~\cite{DBLP:conf/rta/Baader91} 
it was shown that the theory \ACUI of an associative, commutative, and idempotent binary function symbol $f$ with a unit $0$, 
which is unitary w.r.t.\ the restricted instantiation preorder~\cite{DBLP:journals/tcs/BaaderB88} for elementary\footnote{%
This means that unification problems may only contain terms built using variables, $0$, and $f$.
} 
unification, is not unitary w.r.t.\ the unrestricted one, and thus must be finitary, infinitary or of type zero.
In~\cite{DBLP:conf/unif/BaaderL16}, this result was strengthened by demonstrating that also type finitary is not possible.
In~\cite{BaaderG25}, we proved that the unification type of \ACUI is actually infinitary w.r.t.\ the unrestricted instantiation preorder. 
We showed the same result for the theory \AC of an associative and commutative binary function symbol and for  \ACU, which extends \AC with a unit. 
Note that \ACU is unitary and \AC is finitary
for elementary unification w.r.t.\ the restricted instantiation preorder~\cite{DBLP:journals/jacm/Stickel81} (see also Section~10.3 
in~\cite{DBLP:books/daglib/0092409}).
Quite surprisingly, we were also able to show that the unification type of the description logic \EL actually improves 
from type zero~\cite{DBLP:journals/corr/abs-1006-2289} to infinitary when switching from the restricted instantiation preorder to the unrestricted one.
In addition to these results for specific theories/logics, we established some general results on the relationship between the two instantiation 
preorders, which among other things imply that for associativity \Assoc and for commutativity \Commu the unification type does not depend on which of 
the two instantiation preorders is employed.

In the present article, we determine the unrestricted unification type of considerably more equational theories than in~\cite{BaaderG25}.
For instance, the phenomenon that the unification type can improve from type zero to infinitary when going from the restricted instantiation preorder 
to the unrestricted one is exemplified not only with \EL, but also with the theory \AI of an associative and idempotent binary function symbol,
the theory \ACUh, which extends \ACU by finitely many unary function symbols that behave like homomorphisms, and the description logic \FLzero.
We also add several examples of theories that are unitary both in the restricted and in the unrestricted setting, which were first presented in \cite{BaGi-UNIF2026}: 
the theory of Boolean rings \Boole, the theory of Abelian groups \AG, and an extension of the theory of Abelian groups called \GAUSS.
The modal logic \ModK, which is a syntactic variant of the description logic \ALC, provides us with an example where both the restricted and the
unrestricted unification type is zero. We also show that the theory \D, which axiomatizes both-side distributivity, is infinitary both in the restricted
and the unrestricted setting. An overview of our results on specific theories can be found in \autoref{results:table}.
However, the main improvement achieved in the present paper is that, unlike in~\cite{BaaderG25}, we prove these results not separately for each 
theory, but instead are able to establish results for whole classes of theories, of which the theories mentioned above are instances.

This article is structured as follows. 
In the next section, we introduce the specific equational theories as well as the classes of equational theories that are considered in this paper. 
In Section~\ref{basic:sect}, we recall the relevant notions from unification theory. In particular, we introduce the restricted and the unrestricted
instantiation preorder and define unification types w.r.t.\ them. This section also contains the table (\autoref{results:table}) that gives
an overview of our results on specific theories.
Section~\ref{general:results:sect} considers cases where the unification type does not change.
Section~\ref{lbounds:sect}, is dedicated to showing “at least infinitary” lower bounds for the unrestricted unification type of a number of regular theories.
On the one hand, it uses the fact that theories whose restricted unification type is infinitary or zero cannot have unrestricted unification type unitary 
or finitary. On the other hand, it employs a tool for showing “at least infinitary” for the unrestricted unification type of regular theories that
had been introduced in~\cite{DBLP:conf/unif/BaaderL16}.
In Section~\ref{upper:bounds:sec}, we prove “at most infinitary” upper bounds for several classes of regular equational theories.
Section~\ref{conclus:sect} summarizes the results of this article, discusses their usefulness and implications for practice, and mentions interesting directions for future research.

\section{Equational theories}\label{classes:equational:theories:sec}

Given a \emph{signature} $\Sigma$ consisting of a finite set of function symbols (with associated arities) and a countably infinite set
of \emph{variables} $V$, the set $T(\Sigma,V)$ of \emph{terms} over $\Sigma$ with variables in $V$ is defined in the usual way: every variable is a
term, and if $f$ is an $n$-ary function symbol and $t_1,\ldots,t_n$ are terms, then $f(t_1,\ldots,t_n)$ is also a term.
Function symbols of arity $0$ are called \emph{constants}, and we write $c$ rather than $c()$ for the term built using a constant $c$.
For a term $t$, we denote with $\Var(t)$ the set of variables (i.e., elements of $V$) occurring in~$t$. The \emph{size} of $t$
is the number of occurrences of function symbols and variables in $t$.
 
An equational theory $E$ is given by a set of identities $s\id t$ between terms, which are (implicitly) assumed to be universally quantified. 
Such a set of identities $E$ induces the congruence relation $\id_E$ on terms, which can either be defined syntactically through
rewriting or semantically through first-order interpretations of $\Sigma$, with $\id$ as identity relation~\cite{DBLP:books/daglib/0092409}.
We will assume in this paper that all equational theories under consideration are non-trivial in the sense that $x\id_E y$ does not hold for distinct
variables $x, y$. While a trivial equational theory is not inconsistent in the logical sense since it has models of cardinality $1$, it is
similarly degraded since $\id_E$ has a single equivalence class consisting of all terms. Making this general non-triviality assumption
simplifies the formulation of some of our results.

In this section, we introduce several specific equational theories as well as some classes of equational theories that have been
investigated in the context of unification, and which will play an important rôle in this article.

\subsection{Some common equational theories}
We start with presenting some identities that axiomatize frequently occurring properties of mathematical functions. 
Let $f$ be a binary function symbol, and $x,y,z\in V$. Then the following identities respectively state associativity ($\Assoc$), commutativity ($\Commu$),
idempotency ($\Idem$) of $f$, as well as the fact that $f$ has the constant $\unit$ as a unit ($\Unit$):
\[
\begin{array}{l@{\ \ \ \ \ }l}
\Assoc := \{f(x,f(y,z)) \id f(f(x,y),z)\}, & \Commu := \{f(x,y) \id f(y,x)\},\\[.3em] 
\Idem := \{f(x,x) \id x\},                 & \Unit := \{f(x,\unit) \id x\}.
\end{array}
\]
For function symbols satisfying several of these properties, we denote the theory obtained as the union of some of these theories by concatenating
the corresponding letters. For example, $\AC = \Assoc \cup \Commu$ and \ACUI denotes the union of all four theories. By a slight abuse of notation, we
will use the respective letters also for binary function symbols different from $f$ and constants different from $\unit$, but in the union we assume
that the same binary function symbol is used. For example, we may use \AU also for $\{g(x,g(y,z)) \id g(g(x,y),z),g(x,1)\id x\}$, but
\AC not for $\{f(x,f(y,z)) \id f(f(x,y),z), g(x,y)\id g(y,x)\}$ if $f \neq g$. 

Another common equational theory is the theory \Boole of \emph{Boolean rings}. Let $\zero$ and $\one$ be two constants, $\bRingSum$ and $\bRingMult$ two binary functions symbols, and $x,y,z\in V$. The theory \Boole consists of the following set of identities (for better legibility, we use infix notation for $\bRingSum$ and $\bRingMult$):

 \begin{equation*} 
	\Boole:=
	\left\{\! \begin{array}{rllcrll}
		x \bRingSum y & \id & y \bRingSum x & & x \bRingMult y & \id & y \bRingMult x\\[.3em]
		x \bRingSum (y \bRingSum z) & \id & (x \bRingSum y) \bRingSum z & & x \bRingMult (y \bRingMult z) & \id & (x \bRingMult y) \bRingMult z\\[.3em]
		x \bRingSum x & \id & \zero & & x \bRingMult x & \id & x\\[.3em]
		\zero \bRingSum x & \id & x & & \zero \bRingMult x & \id & \zero \\[.3em]
		x * (y \bRingSum z) & \id & (x \bRingMult y) \bRingSum (x \bRingMult z) & & \one \bRingMult x & \id & x
	\end{array}\!
	\right\}.
\end{equation*}
The identities in the first two rows state that $\bRingSum$ and $\bRingMult$ are both commutative and associative function symbols. In the third row, $x \bRingSum x \id \zero$ states that every element is an inverse of itself w.r.t.\ $\bRingSum$, while $x \bRingMult x \id x$ expresses that $\bRingMult$ is idempotent.
The fifth identity in the first column states that $\bRingMult$ distributes over $\bRingSum$. 
The remaining identities in \Boole state that $\bRingSum$ and $\bRingMult$ have the constants $\zero$ and $\one$ as units, respectively, and that the constant $\zero$ is an annihilator for $\bRingMult$.

\subsection{Regular and finite theories}
	The equational theory $E$ is called
	\begin{itemize}
		\item 
		\emph{regular} if $\Var(s)=\Var(t)$ holds for all identities $s\id t$ in $E$.
		\item 
		\emph{finite} if every equivalence class of the induced congruence relation $\id_E$ is finite. That is, for all terms $t \in T(\Sigma,V)$, the set $\{s \in T(\Sigma,V) \mid s \id_E t\}$ is finite.
	\end{itemize}

The following two properties of these classes of theories will turn out to be useful later on.
First, it is well-known that regularity of the defining set of identities of an equational theory implies regularity of the whole theory 
(see, e.g., \cite{DBLP:conf/rta/Yelick85} for a proof). 
\begin{lem}\label{regular:implies:all:regular:lem}
	If $E$ is regular, then $\Var(s)=\Var(t)$ holds for all terms $s, t$ satisfying $s\id_E t$.
\end{lem}
Second, as shown in Lemma~3.3.1 of \cite{BurckertHS89}, finite theories are always regular.
\begin{lem}\label{finite:is:regular:cfree:lem}
	Every finite equational theory is regular.
\end{lem}

Obviously, the identities in \Assoc, \Commu, \Idem, and \Unit satisfy the regularity condition, and thus any equational theory obtained as a union
of these theories is also regular. However, the identities in \Idem and \Unit destroy finiteness. In fact, the chain of equivalences
\[
   x \id_\Idem f(x,x) \id_\Idem f(x,f(x,x)) \id_\Idem f(x,f(x,f(x,x))) \id_\Idem \cdots 
\]
shows that the $\id_\Idem$-equivalence class of the variable $x$ is infinite, and so does the following chain for its $\id_\Unit$-equivalence class:
\[
x \id_\Unit f(x,\unit) \id_\Unit f(x,f(\unit,\unit)) \id_\Unit f(x,f(\unit,f(\unit,\unit))) \id_\Unit \cdots .
\]

Let us continue by presenting examples of finite theories. First, note that the empty theory $\emptyset$ as well as the theories $\Assoc$, $\Commu$, and $\AC$
are obviously finite since they preserve the size and the symbols of terms, i.e., every term in the equivalence class of the term $t$ has the same size as $t$
and contains the same function symbols and variables as $t$.
Second, we consider distributivity of a binary function symbol $m$ (for multiplication) over a binary function symbol $p$ (for plus):
\[
\D := \{
m(z,p(x,y)) \id p(m(z,x),m(z,y)),\ 
m(p(x,y),z) \id p(m(x,z),m(y,z))
\}.
\]
Finiteness of \D was first observed in Szabo's PhD thesis~\cite{Szabo83}, though Szabo calls these theories permutative. 
Since the thesis is written in German and very hard to access, we give a new proof of this fact here.
We associate with $m$ the polynomial $P_m(X,Y) := XY$, with $p$ the polynomial $P_p(X,Y) := X+Y$, and with every constant $c$ the polynomial $P_c := 2$. This basically means that we interpret $m$ as multiplication and $p$ as addition. Let $\pi_2(t)$ be the natural number obtained by evaluating $P_t$ with $2$ substituted for all indeterminates of $P_t$. It is shown in the proof of Lemma~5.1 of~\cite{DBLP:journals/lmcs/BaaderG25} that the size of a given term $t$ is strictly smaller than $\pi_2(t)$.\footnote{%
	Strictly speaking, this result is formulated in~\cite{DBLP:journals/lmcs/BaaderG25} for the nu-size, where unary function symbols are not counted, but since in our setting there are no unary function symbols, the nu-size coincides with the size of the term. 
}
When computing $\pi_2(t)$, the function symbol $m$ is interpreted as multiplication of natural numbers and $p$ is interpreted as addition of natural numbers. Since multiplication distributes over addition from both sides, this means that $s\id_\D t$ implies that $\pi_2(s) = \pi_2(t)$. Consequently, all elements $s$ of the $\id_\D$ equivalence class of $t$ satisfy $|s| < \pi_2(t)$. Since $\pi_2(t)$ is a finite natural number and \D preserves the symbols of $t$, this implies that the $\id_\D$ equivalence class of $t$ is finite.

Summing up, we have shown the following result.

\begin{prop}\label{finite:theories:prop}
	The equational theories $\emptyset$, $\Commu$, $\Assoc$, \AC, and \D are finite.
\end{prop}

\subsection{Monoidal theories}\label{monoidal:theories:def:sec} 

The equational theories \ACU and \ACUI belong to the class of monoidal theories. This class of theories was independently introduced by Nutt \cite{Nutt90} and Baader (as \emph{commutative theories}) \cite{DBLP:journals/jsc/Baader89}, and the correspondence between the two formulations was established in \cite{DBLP:conf/rta/BaaderN91,BaaderN96}.
For simplicity, we use the definition of monoidal theories given in \cite{Nutt90}, since it requires less auxiliary definitions

Let $E$ be an equational theory. The signature of $E$, denoted as $\Sig(E)$, is the set of function symbols occurring in the identities of $E$. 
An equational theory $E$ is called \emph{monoidal} if it satisfies the following properties:
	\begin{itemize}
		\item
		The signature $\Sig(E)$ contains a constant symbol $\unit$ and a binary function symbol $f$, and all other function symbols in $\Sig(E)$ are unary.
		\item
		The symbol $f$ is (modulo $E$) associative and commutative with unit $\unit$. That is,
		\[
		f(x,y) \id_E f(y,x),\ \ f(x,(f(y,z)) \id_E f(f(x,y),z),\ \ f(x,\unit) \id_E x.
		\]
		\item
		Every unary function symbol $h$ in $\Sig(E)$ is a homomorphism for $f$ and $\unit$, i.e.,
		\[
		h(f(x,y)) \id_E f(h(x),h(y)),\ \ h(\unit) \id_E \unit.
		\]
	\end{itemize}
The equational theory \ACU is the least monoidal theory in the sense that its signature contains no unary function symbol and $f$ is not required to satisfy any additional property. 
Another example of a monoidal theory is \ACUI, as is its extension with a homomorphism:
\[
   \ACUIh := \ACUI \cup \{h(f(x,y)) \id f(h(x),h(y)),\ h(\unit) \id \unit\}.
\]
The same is true for the extension \ACUh of \ACU with a homomorphism.
By a slight abuse of notation, we use \ACUIh and \ACUh also to denote the extension of \ACUI and \ACU with more than one homomorphism.
Finally, note that the theory \AG of \emph{Abelian groups} is also a monoidal theory. It is obtained by adding to \ACU a unary function symbol $\inv$ that expresses the existence of an inverse. That is,
\[
   \AG := \ACU \cup \{f(x,\inv(x)) \id \unit\}. 
\]
Note that, although not explicitly stated in \AG, the unary function symbol $\inv$ is also a homomorphism. It is not hard to show that the following holds:
\[
  \inv(f(x,y)) \id_\AG f(\inv(x),\inv(y))\ \ \  \mbox{and}\ \ \ \inv(\unit) \id_\AG \unit.
\] 
Hence, \AG is indeed a monoidal theory. 
An additional example of a monoidal theory is the equational theory \GAUSS \cite{Nutt90}, which can be obtained by extending $\ACUh$ in the following way: 
\[
\GAUSS := \ACUh \cup \{f(x,h(h(x))) \id \unit\}.
\]
In contrast to \ACU, \ACUI and \ACUIh, the theories \AG and \GAUSS are examples of monoidal theories that are not regular. In fact, the identities $f(x,\inv(x)) \id \unit$ in $\AG$ and $f(x,h(h(x))) \id \unit$ in \GAUSS do not satisfy the regularity condition.
Also note that all five of these theories contain \Unit, whose congruence relation $\id_\Unit$ has infinite equivalence classes (e.g., the $\id_\Unit$-equivalence class of the variable $x$, as explained above). Therefore, none of the theories \ACU, \ACUI, \ACUIh, \AG, and \GAUSS is finite.

\subsection{Equational theories induced by description and modal logics} %
In this section, we briefly introduce the description logics \ALC, \EL, and \FLzero and the modal logic \ModK, which is a syntactic variant of \ALC \cite{DBLP:conf/ijcai/Schild91}, and recall the equational theories that can be used to axiomatize equivalence in these logics. More information about these logics can be found in~\cite{DBLP:books/daglib/0041477} for \ALC and \EL, in \cite{DBLP:journals/jsc/BaaderN01} for \FLzero, and in~\cite{DBLP:books/cu/BlackburnRV01} for \ModK. An overview on results for unification in modal and description logics can be found in~\cite{DBLP:journals/igpl/BaaderG11}.

Given disjoint sets of concept names (unary predicates) and role names (binary predicates), \emph{\ALC concept descriptions} (or simply \emph{concepts}) 
are built from concept names using the concept constructors \emph{top} ($\top$), \emph{conjunction} ($C\sqcap D$), \emph{disjunction} ($C\sqcup D$),
\emph{negation} ($\neg C$), \emph{value restriction} ($\forall r.C$), and 
\emph{existential restriction} ($\exists r.C$) according to the following grammar:
\[
	C  ::= \top \mid \: \bot \: \mid \: A \: \mid \: C \sqcap C \: \mid \: C \sqcup C \: \mid \: \neg C \: \mid\forall r.C \: \mid\exists r.C,
\]
where $A$ stands for concept names, $r$ for role names, and $C$ for \ALC concepts. 
\EL is the sublogic of \ALC that uses only the constructors top, conjunction, and existential restriction, 
whereas \FLzero uses only the constructors top, conjunction, and value restriction.

In the model-theoretic semantics of \ALC, a given interpretation \Imc assigns sets $C^\Imc$ to
concept descriptions $C$ according to the semantics of the constructors. 
\begin{defi}\label{sem:alc:def}
	An \emph{interpretation} $\Imc = (\Delta^\Imc,\cdot^\Imc)$
	consists of a non-empty interpretation domain $\Delta^\Imc$ and an extension function $\cdot^\Imc$ that assigns subsets of this domain to concept names and
	binary relations on the domain  to role names. 
	This interpretation function is extended to concept descriptions as follows:
	\[
	\begin{array}{rcl}
		\top^\Imc &:=& \Delta^\Imc,\ \ \bot^\Imc\ \, :=\ \, \emptyset,\\[.3em]
		(C\sqcap D)^\Imc &:=&  C^\Imc\cap D^\Imc,\ \ (C\sqcup D)^\Imc\ \,  :=\ \,  C^\Imc\cup D^\Imc,\ \ (\neg C)^\Imc\ \,  :=\ \,  \Delta^\Imc\setminus C^\Imc,\\[.3em]
		(\exists r.C)^\Imc &:=& \{d\in \Delta^\Imc \mid \mbox{there exists}\ e\in \Delta^\Imc\ \mbox{such that}\  (d,e)\in r^\Imc\ \mbox{and}\  e\in C^\Imc\}, \\[.3em]
		(\forall r.C)^\Imc &:=& \{d\in \Delta^\Imc \mid \mbox{for all}\ e\in \Delta^\Imc, (d,e)\in r^\Imc\ \mbox{implies}\ e \in C^\Imc\}. 
	\end{array}
	\]
\end{defi}

Given two \ALC concepts $C$ and $D$, we say that $C$ is \emph{subsumed} by $D$ (written $C\sqsubseteq D$) if 
$C^\Imc\subseteq D^\Imc$ holds for all interpretations $\Imc$, and that $C$ and $D$ are \emph{equivalent} (written $C\equiv D$) 
if they subsume each other, i.e., $C^\Imc = D^\Imc$ holds for all interpretations $\Imc$.
Since \EL and \FLzero are sublogics of \ALC, this also defines subsumption and equivalence for \EL and \FLzero concept descriptions.

Equivalence of \EL concept descriptions can be axiomatized by the equational theory \bSLmO of \emph{bounded semilattices with 
	monotone operators} \cite{DBLP:journals/cuza/Sofronie-Stokkermans13,DBLP:journals/ndjfl/BaaderNBM16,DBLP:journals/fuin/Sofronie-Stokkermans17}. 
For this purpose, we view the conjunction operator $\sqcap$ as a binary function symbol (written infix), $\top$ as a constant symbol, and
$\exists r.$ for a role name $r$ as a unary function symbol. The theory  
\bSLmO over this signature then consists of the identities stating that $\sqcap$ is associative, commutative, and idempotent, has $\top$ as unit, 
and existential restrictions as monotone operators:
\begin{eqnarray*}
	\bSLmO &:=& \{x\sqcap y \id y\sqcap x,\  (x\sqcap y)\sqcap z\id x\sqcap (y\sqcap z),\  x\sqcap x \id x,\  x\sqcap \top \id x\}\ \cup \\[.2em]
	 && \{\exists r.x\sqcap\exists r.(x\sqcap y) \id \exists r.(x\sqcap y) \mid \mbox{for each role name } r\}.
\end{eqnarray*}
To axiomatize equivalence of \FLzero concepts, we also view $\forall r.$ for a role name $r$ as a unary function symbol. Note that the semantics of value restriction yields the following equivalences: 
\[
  \forall r.(C \sqcap D) \equiv \forall r.C \sqcap \forall r.D \ \ \ \mbox{and} \ \ \ \forall r.\top \equiv \top.
\] 
Obviously, these equivalences tell us that the unary function symbol $\forall r.$ behaves as a homomorphism for $\sqcap$ and $\top$. 
Hence, equivalence of \FLzero concepts can be axiomatized by the monoidal theory \ACUIh, with as many homomorphisms as there are role names~\cite{DBLP:journals/jsc/BaaderN01}.

In the following, we use $\idEL$ and $\idFLzero$ (rather than $\id_\bSLmO$ and $\id_\ACUIh$) to denote the congruence relations induced
by the sets of identities axiomatizing equivalence in \EL and~\FLzero, respectively, and by a slight abuse of notation also use
\EL and~\FLzero to denote these identities.

Equivalence in \ALC can be axiomatized by the identities axiomatizing Boolean algebras for conjunction, disjunction, negation, $\top$, and $\bot$, 
together with identities stating that $\forall r.$ acts as homomorphism on conjunction and $\top$, and the usual definition of existential restrictions
using value restrictions and negation:
$
\exists r.x \id \neg\forall r.\neg x.
$

For the (multi-) modal logic \ModK, which is a syntactic variant of \ALC via the translation
\[
\top\Rightarrow \top, \ \
\bot\Rightarrow \bot,\ \
\neg\Rightarrow \neg,\ \
{\sqcap}\Rightarrow {\wedge},\ \
{\sqcup}\Rightarrow {\vee},\ \
{\forall r.}\Rightarrow \Box_r,\ \
{\exists r.}\Rightarrow \Diamond_r,
\]
this was, e.g., shown in~\cite{DBLP:journals/jsyml/Lemmon66}.
Again, we will use \ModK also to denote the corresponding set of identities and $\id_\ModK$ for the induced congruence relation.

\section{Unification and unification types}\label{basic:sect}

In unification, the goal is to make pairs of terms equivalent w.r.t.\ $\id_E$ for an equational theory $E$ by replacing variables with terms.
Formally, this replacement is achieved by applying a substitution.

A \emph{substitution} $\sigma$ is a mapping from $V$ to $T(\Sigma,V)$ that has a finite \emph{domain} $\Dom(\sigma) := \{x\in V \mid \sigma(x)\neq x\}$.
It can be homomorphically extended to a mapping from $T(\Sigma,V)$ to $T(\Sigma,V)$ by defining 
$\sigma(f(t_1,\ldots,t_n)) := f(\sigma(t_1),\ldots,\sigma(t_n))$. 
The \emph{variable range} $\VRan(\sigma)$ of $\sigma$ consists of the set of variables occurring in the terms $\sigma(x)$ for $x\in\Dom(\sigma)$.
Substitutions can be compared using the instantiation preorder:
given an equational theory $E$, a set of variables $X\subseteq V$, and two substitutions $\sigma, \tau$, we say that
$\sigma$ is \emph{more general} than $\tau$ (or $\tau$ is an \emph{instance} of $\sigma$) w.r.t.\ $E$ and $X$ 
(written $\sigma\leq_E^X \tau$)
if there is a substitution $\lambda$ such that $\lambda\sigma \id_E^X \tau$, i.e., $\lambda(\sigma(x)) \id_E \tau(x)$ holds for all $x\in X$.
In case $X = V$ we also write $\lambda\sigma \id_E \tau$ in place of $\lambda\sigma \id_E^V \tau$.
It is easy to see that $\leq_E^X$ is indeed a preorder, i.e., reflexive and transitive, but in general not antisymmetric.
We write $\sim_E^X$ for the equivalence relation induced by $\leq_E^X$, i.e., $\sigma \sim_E^X \tau$ iff
$\sigma \leq^X_E \tau$ and $\tau \leq^X_E \sigma$.
We say that $\sigma$ is \emph{strictly more general} than $\tau$ (or $\tau$ is a \emph{strict instance} of $\sigma$) w.r.t.\ $E$ and $X$
(written $\sigma <_E^X \tau$) if $\sigma\leq_E^X \tau$ and $\sigma \not\sim_E^X \tau$.

An \emph{$E$-unification problem} is a finite set of equations of the form
\[
\Gamma = \{s_1\id_E^?t_1,\ldots,s_n\id_E^?t_n\}
\] 
where $s_1,t_1,\ldots,s_n,t_n$ are terms in $T(\Sigma,V)$.
An \emph{$E$-unifier} of $\Gamma$ is a substitution $\sigma$ that solves all the equations in $\Gamma$, i.e.,
satisfies $\sigma(s_i) \id_E \sigma(t_i)$ for all $i, 1\leq i\leq n$. 
The unification problem $\Gamma$ is \emph{solvable} if it has an $E$-unifier.
The set of all $E$-unifiers of $\Gamma$ is denoted as $\Umc_E(\Gamma)$.
For \emph{elementary $E$-unification} it is assumed that $\Sigma$ (and thus also $\Gamma$) contains only function symbols
occurring in $E$. For \emph{$E$-unification with constants}, $\Sigma$ and $\Gamma$ may contain additional constant symbols,
and for \emph{general $E$-unification}, $\Sigma$ and $\Gamma$ may contain additional function symbols of arbitrary arity.

\medskip
\noindent
\textbf{General assumption:} \emph{In this paper, we restrict our attention to elementary $E$-unification.}

\medskip

Unification types for non-empty sets of identities $E$ are usually defined w.r.t.\ the \emph{restricted instantiation preorder}, 
which is $\leq_E^X$ where $X$ is the finite
set $\Var(\Gamma)$ of all variables occurring in the given unification problem $\Gamma$, but some authors also use the 
\emph{unrestricted instantiation preorder} $\leq_E^V$ where $V$ is the countably infinite set of variables available for building terms. 
In this paper, we will additionally consider settings where $X$ is between these two extremes. Note that $\Var(\Gamma)\subseteq X$ is required for the set of $E$-unifiers to be closed under instantiation.

\begin{lem}
If $\Gamma$ is an $E$-unification problem and $X\subseteq V$ a set of variables satisfying $\Var(\Gamma)\subseteq X$, then
$\sigma\in\Umc_E(\Gamma)$ implies $\theta\in\Umc_E(\Gamma)$ for all substitutions $\theta$ such that $\sigma\leq_E^X \theta$.
\end{lem}

Given an $E$-unification problem $\Gamma$ and some set of variables $X$ with $\Var(\Gamma)\subseteq X\subseteq V$, we say that
a set $\Smc$ of substitutions is a \emph{complete set} of $E$-unifiers of $\Gamma$ w.r.t.\ $\leq_E^X$ if it consists of $E$-unifiers of $\Gamma$, 
and every $E$-unifier of $\Gamma$ is an instance of an element of the complete set, i.e., for every $\theta\in\Umc_E(\Gamma)$ there exists
$\sigma\in\Smc$ such that $\sigma \leq_E^X \theta$. 
Such a set is called \emph{minimal} if it does not contain two distinct elements that are comparable w.r.t.\ $\leq_E^X$. 
It is easy to see that minimal complete sets of $E$-unifiers of a given unification problem $\Gamma$ are unique up to the equivalence relation 
$\sim_E^X$ induced by the preorder $\leq_E^X$ (see, e.g., Corollary 3.13 in~\cite{DBLP:books/el/RV01/BaaderS01} and \autoref{th:min} below), 
and thus all have the same cardinality.

\begin{defi}
Let $\Gamma$ be a solvable $E$-unification problem and $X$ a set of variables such that $\Var(\Gamma)\subseteq X\subseteq V$. 
Then the \emph{unification type} of $\Gamma$ w.r.t.\ $\leq_E^X$ is
\begin{itemize}
\item
  \emph{unitary} if $\Gamma$ has a minimal complete set of $E$-unifiers of cardinality one w.r.t.\ $\leq_E^X$, whose
   single element is then called \emph{most general $E$-unifier (mgu)},
\item
  \emph{finitary} if $\Gamma$ has a finite minimal complete set of $E$-unifiers of cardinality greater than one w.r.t.\ $\leq_E^X$,
\item
  \emph{infinitary} if $\Gamma$ has an infinite minimal complete set of $E$-unifiers w.r.t.\ $\leq_E^X$, 
\item
  \emph{zero} if $\Gamma$ does not have a minimal complete set of $E$-unifiers w.r.t.\ $\leq_E^X$, i.e.,
  every complete set is redundant in the sense that it must contain two distinct elements that are comparable w.r.t.\ $\leq_E^X$.
\end{itemize}
\end{defi}
 
Minimal complete sets of unifiers can alternatively be characterized using the following order-theoretic point of 
view~\cite{DBLP:conf/rta/Baader89,DBLP:books/el/RV01/BaaderS01}. 
Let $\Gamma$ be an $E$-unification problem and $X\subseteq V$ a set of variables satisfying $\Var(\Gamma)\subseteq X$.
We denote the $\sim_E^X$-equivalence class of a unifier $\sigma$ as $[\sigma]_E^X$ and the set of all equivalence classes of unifiers 
as $[\Umc_E(\Gamma)]_E^X$. The partial order $\preceq_E^X$ on $[\Umc_E(\Gamma)]_E^X$ induced by the instantiation preorder $\leq_E^X$ 
on unifiers is defined as $[\sigma]_E^X\preceq_E^X[\tau]_E^X$ if $\sigma\leq_E^X\tau$.
We say that $S\subseteq [\Umc_E(\Gamma)]_E^X$ is complete w.r.t.\ $\preceq_E^X$
if every element of $[\Umc_E(\Gamma)]_E^X$ is above (w.r.t.\ $\preceq_E^X$) some element of $S$.

       \begin{thm}[\cite{DBLP:books/el/RV01/BaaderS01}]
       \label{th:min}
Let $M$ be the set of $\preceq_E^X$-minimal elements of $[\Umc_E(\Gamma)]_E^X$.
If $\Smc$ is a minimal complete set of $E$-unifiers of $\Gamma$ w.r.t.\ $\leq_E^X$, then $M=\{[\sigma]_E^X\mid\sigma\in \Smc\}$.
Conversely, if $M$ is complete in $[\Umc_E(\Gamma)]_E^X$, then any set of substitutions obtained by picking one representative for each element
of $M$ is a minimal complete set of $E$-unifiers of $\Gamma$.
       \end{thm}
   
Consequently, unification type zero corresponds to the case where the set $M$ of minimal elements is not complete, whereas 
the other types are determined by the cardinality of the set $M$ in case it is complete.
Theorem~3.1 in~\cite{DBLP:conf/rta/Baader89} establishes  conditions that are necessary, sufficient, or both for proving
unification type zero of a given unification problem.\footnote{%
        Note that, in~\cite{DBLP:conf/rta/Baader89}, the instantiation preorder is written the other way round, i.e., more general substitutions
        are larger.
}

As usual, we order unification types w.r.t.\ how bad they are (larger is worse) by setting
\begin{center}
	\textrm{zero} $>$ \textrm{infinitary} $>$ \textrm{finitary} $>$ \textrm{unitary}.
\end{center}
Based on this order, we define the different notions of restricted and unrestricted unification type of an equational theory.
 \begin{defi}
 	Let $E$ be an equational theory and $X \subseteq V$. Then,
 	\begin{itemize}
 		\item 
 		The unification type of $E$ w.r.t.\ $\leq_E^X$ is the worst type of any solvable $E$-unification problem w.r.t.\ $\leq_E^X$.
 		\item 
 		The \emph{unrestricted unification type} of $\Gamma$ ($E$) is the one w.r.t.\ $\leq_E^V$. 
 		\item 
 		The \emph{restricted unification type} of $\Gamma$ ($E$) is the one w.r.t.\ $\leq_E^X$ for $X=\Var(\Gamma)$.
 	\end{itemize}
 \end{defi}
\begin{table}
 \centering
 	\begin{tabular}{c|c|c|c|c}
 		\hline
 		\diagbox{restricted}{unrestricted} & unitary & finitary & infinitary & zero \\[.3em]
 		\hline
 		unitary & $\emptyset$, \Boole, \AG, \GAUSS & ? & \ACU, \ACUI& ? \\[.3em]
 		\hline
 		finitary & \diagbox[dir=NE] & \Commu, \Idem & \AC, \ACI& ? \\[.3em]
 		\hline
 		infinitary & \diagbox[dir=NE] & \diagbox[dir=NE] & \Assoc, \D & ?\\[.3em]
 		\hline
 		zero & \diagbox[dir=NE] & \diagbox[dir=NE] & \EL, \AI, \FLzero, \ACUh & $\ModK$ \\[.3em]
 		\hline
 	\end{tabular}
 	\caption{Overview on results about the restricted and unrestricted unification type for elementary unification.}
 	\label{results:table}
 \end{table}
 
\medskip 

The results on the unification type of equational theories in the literature 
(see, e.g., \cite{JoKi91,DBLP:books/ox/LAI2Hb1994/BaaderS94,DBLP:books/el/RV01/BaaderS01} for overviews)
are usually shown for the restricted case. 
As we will demonstrate in this article, it may indeed make a considerable difference for the unification type which instantiation preorder is
employed in its definition. Table~\ref{results:table} summarizes our main results, where in all but one case the restricted unification type was known and
the unrestricted one is determined in this paper. For $\ModK$, it is the other way round. As this table shows, when going from the restricted to the unrestricted
instantiation preorder, the unification type can both get worse (e.g., for \ACU) and get better (e.g., for \EL), but it may also stay the same (e.g., for \D).

The struck out cells of the table indicate changes that are not possible due to the fact that ${\leq_E^{V}}\subseteq {\leq_E^{\Var(\Gamma)}}$.
The contrapositive of the following theorem justifies the struck out cells. 

\begin{thm}\label{typ:no:worse:thm}
	Let $E$ be an equational theory. If $E$ has unrestricted unification type unitary (finitary), then it has
	restricted unification type unitary (finitary or unitary).
\end{thm}

\begin{proof}
	A given finite minimal complete set \Smc of $E$-unifiers of $\Gamma$ w.r.t.\ $\leq_E^{V}$ is also complete w.r.t.\ $\leq_E^{\Var(\Gamma)}$.
	If it is not minimal (which can only happen in the finitary case), then one can make it minimal by removing redundant elements
	(i.e., elements $\theta$ such that \Smc contains an element $\sigma <_E^{\Var(\Gamma)} \theta$)
	without destroying completeness.
	%\qed
\end{proof}
The cells containing a question mark indicate cases which are (at least from an order-theoretic point of view) possible, but for which we do
not have concrete equational theories as examples.

\medskip

The main technical sections of this paper (Sections~\ref{lbounds:sect} and~\ref{upper:bounds:sec}) prove the results on the unrestricted unification
type of Table~\ref{results:table} by showing lower bounds (e.g., the unrestricted unification type of \ACU is at least infinitary) and
upper bounds (e.g., the unrestricted unification type of \ACU cannot be zero). 
But first, we consider in the next section cases where the unification type does not change (this concerns the theories $\emptyset$, \Boole, \AG, \GAUSS, \Commu, \Idem, \Assoc, and \ModK).

\section{When the employed instantiation preorder does not matter}\label{general:results:sect}

The purpose of the first result shown in this section is to clarify the reason why the unification type can change when going from
the restricted to the unrestricted instantiation preorder. It is not that the unrestricted preorder also takes variables into account that do not
occur in the unification problem and not even that it considers infinitely many such variables. The reason is that it does not leave infinitely many
variables unobserved.

\begin{prop}\label{types:same:lem}
	Let $E$ be an equational theory, $\Gamma$ an $E$-unification problem, and $X$ a set of variables such that 
	$\Var(\Gamma) \subseteq X \subseteq V$ and $V\setminus X$ is infinite. If $\Gamma$ has a minimal complete set of $E$-unifiers w.r.t.\ $\leq_E^X$, then it has a
	minimal complete set of $E$-unifiers w.r.t.\ $\leq_E^{Var(\Gamma)}$ of the same cardinality, and vice versa.
\end{prop}

\begin{proof}
To simplify the presentation of our proof, we will use $X_0$ to denote the set of variables occurring in $\Gamma$, i.e., $X_0 := \Var(\Gamma)$.

Let $\Smc$ be a minimal complete set of $E$-unifiers of $\Gamma$ w.r.t.\ $\leq_E^{X}$. For every unifier $\theta$ in $\Smc$  
we can rename the variables in $\VRan(\theta)$ such that they do not belong to $X$ by applying an appropriate permutation $\pi$ that
maps the variables of $\VRan(\theta)$ bijectively to a set of variables $Y$ with $|Y| = |\VRan(\theta)|$ and $Y\cap (X\cup\VRan(\theta)) = \emptyset$.
Such a finite set $Y$ of variables exists since $\VRan(\theta)$ is finite and $V\setminus X$ is infinite.
For a given bijection $p : \VRan(\theta)\rightarrow Y$, we can define $\pi$ as follows:
$\pi(z) := p(z)$ for all $z\in \VRan(\theta)$, $\pi(y) := p^{-1}(y)$ for all $y\in Y$, and $\pi(x) = x$ for all other variables.
Note that $\pi$ is a substitution since its domain is $\VRan(\theta)\cup Y$, which is finite.
To show that $\pi$ really is a permutation, i.e., a bijective mapping from $V$ to $V$, it is sufficient to show that it is an injective mapping
from $V$ to $V$ (see Lemma~2.6 in~\cite{DBLP:journals/jsc/Eder85}). This is an immediate consequence of the facts that
$p$ and $p^{-1}$ are bijections and $Y$ and $\VRan(\theta)$ are disjoint.
 
The substitution $\pi\theta$ is equivalent to $\theta$ w.r.t.\ the equivalence relation $\sim_\emptyset^{V}$ induced by 
$\leq_\emptyset^{V}$, and thus also w.r.t.\ $\sim_E^{X}$. In fact, $\pi^{-1}\pi\theta = \theta$. 
If we restrict the domain of this substitution to $X$, then the resulting substitution $\pi\theta|_X$ is still equivalent to
$\theta$ w.r.t.\ $\sim_E^{X}$ and also satisfies $\VRan(\pi\theta|_X) \cap X = \emptyset$.
To see the latter, consider a variable $x\in \Dom(\pi\theta|_X)$.
First assume that $x\in\Dom(\theta)\cap X$. Then all the variables in $\theta(x)$ belong to $\VRan(\theta)$, and thus all the variables in
$\pi\theta|_X(x)$ belong to $Y$, which is disjoint with $X$. If $x\in X\setminus\Dom(\theta)$, then $\pi\theta|_X(x) = \pi(x) \neq x$.
Since $\Dom(\pi) = \VRan(\theta) \cup Y$ and $Y$ is disjoint with $X$, this implies $x\in \VRan(\theta)$, and thus $\pi(x)\in Y$.
 
Let $\Smc'$ be the set of substitutions obtained from $\Smc$ by applying this renaming and domain restriction process to every element of $\Smc$.
Due to the $\sim_E^{X}$-equivalence of the elements of $\Smc$ with their modified variants in $\Smc'$ and the fact that
$X$ contains all the variables occurring in $\Gamma$, it is easy to see that $\Smc'$ is also a minimal complete set of 
$E$-unifiers of $\Gamma$ w.r.t.\ $\leq_E^{X}$. In a second step, we restrict the domains of the elements of $\Smc'$ to $X_0$.
Given an element $\theta$ of $\Smc'$, we define $\theta|_{X_0}$ to be the substitution that coincides
with $\theta$ on $X_0$ and maps each variable $x\not\in X_0$ to $x$. Since $\theta$ is an $E$-unifier of $\Gamma$ and this unification problem
contains only variables from $X_0$, the substitution $\theta|_{X_0}$ is clearly also an $E$-unifier of $\Gamma$.
We claim that $\theta|_{X_0} \leq_E^{X} \theta$.
To prove this, we define the substitution $\lambda$ by setting $\lambda(x) := \theta(x)$ if $x\in X\setminus X_0$ and
$\lambda(x) := x$ for all other variables. We now show that $\lambda\theta|_{X_0} \id_E^X \theta$.
If $x\in X_0\cap\Dom(\theta)$, then $\lambda(\theta|_{X_0}(x)) = \lambda(\theta(x)) = \theta(x)$. 
The first identity holds by the definition of $\theta|_{X_0}$ and the second since the variables occurring in $\theta(x)$ are elements
of $\VRan(\theta)$, and thus do not belong to $X$.
If $x\in X_0\setminus\Dom(\theta)$, then $\lambda(\theta|_{X_0}(x)) = \lambda(x) = x = \theta(x)$.
Finally, if $x\in X\setminus X_0$, then $\lambda(\theta|_{X_0}(x)) = \lambda(x) = \theta(x)$. 
Again, the first identity holds by the definition of $\theta|_{X_0}$ and the second by the definition of $\lambda$.
Summing up, we have shown that the following holds for every element $\theta$ of $\Smc'$: 
$\theta|_{X_0}$ is an $E$-unifier of $\Gamma$ and $\theta|_{X_0} \leq_E^{X} \theta$.
Since $\Smc'$ is a minimal complete set of $E$-unifiers of $\Gamma$ w.r.t.\ $\leq_E^{X}$, its elements are minimal w.r.t.\ $\leq_E^{X}$,
which yields $\theta|_{X_0} \sim_E^{X} \theta$. Consequently, we know that the set
$\Smc'|_{X_0} := \{\theta|_{X_0} \mid \theta\in \Smc'\}$ is also a minimal complete set of $E$-unifiers of $\Gamma$ w.r.t.~$\leq_E^{X}$.

We claim that the same is true w.r.t.\ the smaller set of variables $X_0$, i.e., 
that $\Smc'|_{X_0}$ is also minimal and complete w.r.t.\ $\leq_E^{X_0}$.
Completeness trivially follows from the fact that ${\leq_E^{X}}\subseteq{\leq_E^{X_0}}$.
To prove minimality, assume that $\theta$ and $\tau$ are two distinct elements of $\Smc'|_{X_0}$. Then these
two substitutions are not comparable w.r.t.\ $\leq_E^{X}$. Assume that they are comparable w.r.t.\ $\leq_E^{X_0}$,
i.e., there is a substitution $\lambda$ such that $\lambda(\theta(x)) \id_E \tau(x)$ holds for all $x\in X_0$.
By our construction of $\Smc'|_{X_0}$, we know that $\Dom(\theta)\subseteq X_0$, $\Dom(\tau)\subseteq X_0$, and $\VRan(\theta)\cap X = \emptyset$.
If $x$ is a variable in $X\setminus X_0$, then $\lambda(\theta(x)) = \lambda(x)$ and $\tau(x) = x$.
Thus, if we modify $\lambda$ to $\lambda'$ such that $\lambda'(x) = x$ holds  for all $x\in X\setminus X_0$, then 
$\lambda'(\theta(x)) = \tau(x)$ holds for all $x\in X\setminus X_0$. This modification has no effect on the variables $x\in X_0$.
In fact, let $x$ be such a variable. If $x\in \Dom(\theta)$, then $\theta(x)$ does not contain any variable from $X$, and
thus  $\lambda'(\theta(x)) = \lambda(\theta(x)) \id_E \tau(x)$ holds. If $x\notin \Dom(\theta)$, then $\theta(x) = x$, and thus again
$\lambda'(\theta(x)) = \lambda'(x) = \lambda(x) = \lambda(\theta(x)) \id_E \tau(x)$ since $\lambda'$ coincides with $\lambda$ on the variables in $X_0$.
Summing up, we have shown that the assumption $\theta \leq_E^{X_0} \tau$ implies $\theta \leq_E^{X} \tau$, which contradicts the
fact that $\Smc'|_{X_0}$ is minimal w.r.t. $\leq_E^{X}$. Consequently, $\Smc'|_{X_0}$ is a minimal complete set w.r.t.\ $\leq_E^{X_0}$,
and this set has the same cardinality as the minimal complete set $\Smc$ of $E$-unifiers of $\Gamma$ w.r.t.\ $\leq_E^{X}$ we have started with.

Conversely, let $\Smc$ be a minimal complete set of $E$-unifiers of $\Gamma$ w.r.t.\ $\leq_E^{X_0}$.
By applying the same construction as in the proof of the other direction, we can assume without loss of generality
that every unifier $\theta\in \Smc$ satisfies $\Dom(\theta)\subseteq X_0$ and $\VRan(\theta)\cap X = \emptyset$.
In fact, by applying this construction to a minimal complete set $\Smc'$ of $E$-unifiers of $\Gamma$ w.r.t.\ $\leq_E^{X_0}$,
we obtain a new set $\Smc$ where every element $\theta'$ of the original set $\Smc'$ is replaced with an element $\theta$ that satisfies the above 
conditions and is $\sim_E^X$-equivalent to $\theta'$, and thus also $\sim_E^{X_0}$-equivalent. Consequently, this new set $\Smc$ is also
a minimal complete set of $E$-unifiers of $\Gamma$ w.r.t.\ $\leq_E^{X_0}$.

We claim that $\Smc$ is also minimal and complete w.r.t.\ the larger set of variables $X$.
Minimality trivially follows from the fact that ${\leq_E^{X}}\subseteq{\leq_E^{X_0}}$.
To prove completeness, assume that $\tau$ is an $E$-unifier of $\Gamma$. Completeness of $\Smc$ w.r.t.\ $\leq_E^{X_0}$
yields an element $\theta$ of $\Smc$  such that $\theta \leq_E^{X_0} \tau$, i.e., there is a substitution $\lambda$ such that
$\lambda\theta(x) \id_E \tau(x)$ holds for all $x\in X_0$. 
Let $x$ be a variable in $X \setminus X_0$. Then $\lambda\theta(x) = \lambda(x)$. 
Thus, if we modify $\lambda$ to $\lambda'$ such that $\lambda'(x) = \tau(x)$ holds  for all $x\in X\setminus X_0$, then
$\lambda'(\theta(x)) = \tau(x)$ holds for all $x\in X\setminus X_0$.
The claim that this modification has no effect for the variables $x\in X_0$ can be shown as in the proof of minimality of 
$\Smc'|_{X_0}$  w.r.t.\ $\leq_E^{X_0}$ above. This proves that $\lambda'\theta(x) \id_E \tau(x)$ holds for all $x\in X$,
and thus $\Smc$ is also complete w.r.t.\ $\leq_E^{X}$.
%\qed
\end{proof}

The following theorem is an immediate consequence of this lemma.

\begin{thm}\label{infinite:unobserved:thm}
Let $E$ be an equational theory, $\Gamma$ an $E$-unification problem, and $X$ a set of variables such that
$\Var(\Gamma)\subseteq X \subseteq V$ and $V\setminus X$ is infinite. Then the restricted unification type of $\Gamma$ coincides
with the unification type of $\Gamma$ w.r.t.\ $\leq_E^{X}$.
\end{thm}

Note that the condition of the theorem is in particular satisfied if $X$ is a finite superset of $\Var(\Gamma)$.
However, it is clearly not satisfied for $X = V$, which corresponds to the unrestricted instantiation preorder setting. 
We will see in the  next two sections that in this case the unification type can indeed depend on whether the restricted or the unrestricted
instantiation preorder is used. However, there is a special case where this cannot happen.

\begin{lem}\label{stays:same:lem}
	Let $E$ be an equational theory, $\Gamma$ an $E$-unification problem, and \Smc a set of $E$-unifiers of $\Gamma$ such that
	$\VRan(\sigma)\cup\Dom(\sigma) \subseteq \Var(\Gamma)$ holds for all $\sigma\in\Smc$. 
	Then $\Smc$ is a minimal complete set of $E$-unifiers of $\Gamma$ w.r.t.\ $\leq_E^{V}$
	iff it is a minimal complete set of $E$-unifiers of $\Gamma$ w.r.t.\ $\leq_E^{\Var(\Gamma)}$.
\end{lem}

\begin{proof}
	Let $\Smc$ be a minimal complete set of $E$-unifiers of $\Gamma$ w.r.t.\ $\leq_E^{V}$.
	Since ${\leq_E^{V}}\subseteq{\leq_E^{\Var(\Gamma)}}$, this set is also complete w.r.t.\ $\leq_E^{\Var(\Gamma)}$.
	To show minimality w.r.t.\ $\leq_E^{\Var(\Gamma)}$, 
	assume to the contrary that $\sigma,\theta$ are two distinct elements of \Smc such that $\sigma\leq_E^{\Var(\Gamma)}\theta$,
	i.e., there is a substitution $\lambda$ such that $\lambda\sigma(x) \id_E \theta(x)$ holds for all $x\in\Var(\Gamma)$.
	We modify $\lambda$ to $\lambda'$ by setting $\lambda'(x) = x$ for all variables $x\in V\setminus\Var(\Gamma)$.
	For $x\in\Var(\Gamma)$, we know that $\sigma(x)$ contains only variables from $\VRan(\sigma)\subseteq\Var(\Gamma)$ if $x\in \Dom(\sigma)$ or
	$\sigma(x) = x \in \Var(\Gamma)$. Since $\lambda$ and $\lambda'$ coincide on $\Var(\Gamma)$, this yields
	$\lambda'\sigma(x) = \lambda\sigma(x) \id_E \theta(x)$. 
	For $x\in V\setminus\Var(\Gamma)$, this variable does not belong to any of the sets $\Dom(\sigma)$, $\Dom(\theta)$, and $\Dom(\lambda')$,
	and thus $\lambda'\sigma(x) = \lambda'(x) = x = \theta(x)$. Summing up, we have shown that $\sigma\leq_E^{V}\theta$, which contradicts
	our assumption that $\Smc$ is minimal w.r.t.\ $\leq_E^{V}$.
	
	Conversely, assume that $\Smc$ is a minimal complete set of $E$-unifiers of $\Gamma$ w.r.t.\ $\leq_E^{\Var(\Gamma)}$.
	Since ${\leq_E^{V}}\subseteq{\leq_E^{\Var(\Gamma)}}$, minimality also holds w.r.t.\ $\leq_E^{V}$.
	To show completeness w.r.t.\ $\leq_E^{V}$, assume that $\theta$ is an $E$-unifier of $\Gamma$. Completeness of \Smc w.r.t.\ $\leq_E^{\Var(\Gamma)}$
	yields a substitution $\sigma\in \Smc$ such that $\sigma\leq_E^{\Var(\Gamma)}\theta$. Similarly to the first part of the proof, we can show that
	this also implies $\sigma\leq_E^{V}\theta$. The difference is that now $\theta$ is an arbitrary unifier, and thus
	$\VRan(\theta)\cup\Dom(\theta) \subseteq \Var(\Gamma)$ need not hold.
	Let $\lambda$ be such that $\lambda\sigma(x) \id_E \theta(x)$ holds for all $x\in\Var(\Gamma)$.
	We modify $\lambda$ to $\lambda'$ by setting $\lambda'(x) = \theta(x)$ for all variables $x\in V\setminus\Var(\Gamma)$.
	For $x\in\Var(\Gamma)$, we obtain $\lambda'\sigma(x) = \lambda\sigma(x) \id_E \theta(x)$ as in the first part of the proof.
	For $x\in V\setminus\Var(\Gamma)$, this variable does not belong to $\Dom(\sigma)$, and thus
	$\lambda'\sigma(x) = \lambda'(x) = \theta(x)$.
	%\qed
\end{proof}

Note that the proof of completeness in each direction did not make use of the minimality assumption. Thus, it also shows the following lemma.

\begin{lem}\label{stays:same:lem2}
	Let $E$ be an equational theory, $\Gamma$ an $E$-unification problem, and \Smc a set of $E$-unifiers of $\Gamma$ such that
	$\VRan(\sigma)\cup\Dom(\sigma) \subseteq \Var(\Gamma)$ holds for all $\sigma\in\Smc$. 
	Then $\Smc$ is a complete set of $E$-unifiers of $\Gamma$ w.r.t.\ $\leq_E^{V}$
	iff it is a complete set of $E$-unifiers of $\Gamma$ w.r.t.\ $\leq_E^{\Var(\Gamma)}$.
\end{lem}

The following theorem is an easy consequence of these two lemmas.

\begin{thm}\label{stays:same:thm}
	Let $E$ be an equational theory, $\Gamma$ an $E$-unification problem, and \Smc a set of $E$-unifiers of $\Gamma$ such that
	$\VRan(\sigma)\cup\Dom(\sigma) \subseteq \Var(\Gamma)$ holds for all $\sigma\in\Smc$ and \Smc is complete w.r.t.\ the restricted or
	the unrestricted instantiation preorder. Then the restricted and the unrestricted $E$-unification types of $\Gamma$ coincide.
\end{thm}

\begin{proof} 
	By \autoref{stays:same:lem2}, the assumption that \Smc is complete w.r.t.\ one of the two instantiation preorders implies that it is
	also complete w.r.t.\ the other.
	
	If $\Gamma$ has a minimal complete set of $E$-unifiers w.r.t.\ $\leq_E^{\Var(\Gamma)}$,
	then it also has such a minimal complete set that is contained in \Smc. In fact, for \Smc to be complete, it must contain a representative of 
	each minimal $\sim_E^{\Var(\Gamma)}$-equivalence class, and thus we obtain a minimal complete set contained in \Smc by taking one representative in  \Smc
	for each class. By \autoref{stays:same:lem}, this minimal complete set w.r.t.\ $\leq_E^{\Var(\Gamma)}$ is also minimal complete w.r.t.\ $\leq_E^{V}$.
	Thus, types unitary, finitary, and infinitary transfer from the restricted instantiation preorder to the unrestricted one. In the same way, we can also show that 
	types unitary, finitary, and infinitary transfer from the unrestricted instantiation preorder to the restricted one.
	
	Since type zero is the only other type, it thus follows immediately that it also transfers between the two kinds of orders.
\end{proof}

Examples of equational theories where the conditions of this theorem are always satisfied are the empty theory (unitary),
the theory $\Commu$ (finitary), the theory $\Assoc$ (infinitary), the theory $\Idem$ axiomatizing
idempotency of a binary function symbol (finitary), the theory $\Boole$ of Boolean rings (unitary), the theory \AG of Abelian groups (unitary), and the theory \GAUSS (unitary).

\begin{cor}\label{stays:same:cor}
	For the theories $\emptyset$, $\Commu$, $\Assoc$, $\Idem$, \Boole, \AG, and \GAUSS, the restricted and the unrestricted unification types coincide.
\end{cor}

\begin{proof}
	For the first three theories $\emptyset$, $\Commu$, and $\Assoc$, applicability of \autoref{stays:same:thm} is an easy consequence of the known algorithms
	\cite{DBLP:journals/toplas/MartelliM82,DBLP:conf/eurosam/Siekmann79a,Plotkin72}
	computing (or enumerating, in the case of $\Assoc$) minimal complete sets of unifiers for these theories w.r.t.\ the restricted instantiation preorder.
	
	For idempotency $\Idem$, dedicated algorithms are described in~\cite{DBLP:conf/ijcai/KuhnerMRS77,DBLP:conf/ki/Herold87}. However, the fact that $\Idem$
	is finitary already follows from the fact that the rewrite system consisting of the single rule $f(x,x) \rightarrow x$ is clearly canonical (i.e.,
	confluent and terminating) and that this system satisfies Hullot's termination criterion for basic narrowing. Thus, basic narrowing enumerates a finite, complete
	set of $\Idem$-unifiers w.r.t.\ the restricted instantiation preorder (see Theorem~4, Proposition~1, and Example~1 in~\cite{DBLP:conf/cade/Hullot80}).
	This set satisfies the conditions formulated for the set \Smc in \autoref{stays:same:thm}.
	In fact, both the narrowing steps and the final unification step employ syntactic unification (i.e., unification modulo the theory $\emptyset$), and
	thus the computed substitutions use only variables occurring in the respective unification problem in their domains and variable ranges. Consequently,
	the same holds for their compositions, which are the elements of the complete set.
	
	For the theory \Boole of Boolean rings, applicability of \autoref{stays:same:thm} is an easy consequence of the known algorithms for computing most general unifiers (see, e.g., the description of the algorithm based on Löwenheim's
	formula in \cite{DBLP:books/daglib/0092409}).
	
	As for the theories \AG and \GAUSS, note that both are monoidal theories, for which unification can be reduced to solving linear equations over an associated
	semiring (see \cite{Nutt90,BaaderN96} and Section~\ref{monoidal:theories:sec}). For \AG, this is the ring of integers \Zbb and for \GAUSS it is the ring of Gaussian integers $\Zbb[i]$. Both rings are
	principal ideal domains, which implies that the set of solutions of a system of linear equations with $n$ variables has a generating set consisting of at most
	$n$ elements (see, e.g., Theorem~3.7 in \cite{Jacobson}). Since each of the elements of such a generating set corresponds to a variable in the most general unifier
	in the restricted setting, this implies that no new variables are needed, and thus \autoref{stays:same:thm} applies.
\end{proof} 

Another equational theory for which \autoref{stays:same:thm} turns out to be useful is the theory \ModK. In \cite{DBLP:journals/logcom/Jerabek15}, it was shown by Jer{\'{a}}bek that \ModK has unification type zero w.r.t.\ the \emph{unrestricted} instantiation
preorder. He mentions in the paper that his proof can easily be modified to show this result also for the restricted case, but does not give details.
However, for the unification problem he considers, our \autoref{stays:same:thm} is applicable, and thus yields a formal proof that \ModK is also
of  unification type zero w.r.t.\ the restricted instantiation preorder.
To be more precise, Jer{\'{a}}bek considers the modal formula  $p\rightarrow \Box p$ and wants to make it valid by applying a substitution, i.e., by finding
a modal formula $\phi$ such that $\phi\rightarrow \Box \phi$ is valid according to the sematics of the modal logic \ModK.
Translated into a unification problem of the form investigated in the present paper, this means that he considers the \ModK-unification problem
$\{p\rightarrow \Box p \id_\ModK^? \top\}$, where \ModK now denotes the equational theory axiomatizing equivalence for the modal logic \ModK.
He shows that this unification problem has type zero w.r.t.\ the unrestricted instantiation preorder. As part of the proof, he exhibits a complete set of unifiers (w.r.t.\ the unrestricted instantiation preorder) consisting of infinitely many unifiers whose 
domains and variable ranges are contained in $\{p\}$. 
Thus, our \autoref{stays:same:thm} is applicable with respect to this complete set, which shows that type zero transfers from the unrestricted to the restricted case.

\begin{cor}\label{modal:k:result}
	For the equational theory axiomatizing equivalence in the modal logic \ModK, the restricted and the unrestricted unification types coincide,
        and is thus zero in both cases.
\end{cor}

\input{lower-bounds}

\input{upper-bounds}

\section{Conclusion}\label{conclus:sect}

We have investigated the effect that the employed instantiation preorder has on the unification type of an equational theory.
In this section, we will first review the results obtained during this investigation, then discuss the impact they may have on
applications of equational unification, and finally point out related and possible future work.

\subsection{Our results}
As a rule of thumb, one can extract from our investigation that nothing changes if unifiers in a minimal complete set w.r.t.\ the restricted
instantiation preorder do not need fresh variables (\autoref{stays:same:thm}), whereas the unification type switches from unitary or finitary
to at least infinitary otherwise (\autoref{NUOF:thm}), though the latter result was only shown for regular theories. We have employed \autoref{NUOF:thm}
in Section~\ref{lbounds:sect} to prove that the unification type of the frequently used theories \ACUI, \ACI, \ACU, and \AC is at least infinitary 
w.r.t.\ the unrestricted instantiation preorder (\autoref{ACUI:cor}, \autoref{ACI:cor}, and \autoref{ACU:cor}). 
For theories whose restricted unification type is infinitary or zero, one can use \autoref{typ:no:worse:thm} to show the lower bound ``at least infinitary''
for the unrestricted unification type. In Section~\ref{lbounds:sect} we have employed this argument for the theories \AI, \D, \EL, \FLzero, and \ACUh.
Our main tool for proving ``at most infinitary'' upper bounds for regular theories in Section~\ref{upper:bounds:sec} is \autoref{regular:theory:prop}, 
which basically states that, up to equivalence, substitutions that have a given substitution $\theta$ as instance do not ``use'' more variables than $\theta$.
We have employed this result to prove ``at most infinitary'' upper bounds for several classes of regular theories:
regular finite theories, regular locally finite theories, and regular monoidal theories that satisfy an additional restriction on the associated semiring, 
called restrictiveness.

We have clarified by \autoref{infinite:unobserved:thm}
that the reason for changes in the unification type when going from the restricted to the unrestricted setting
is not that the unrestricted instantiation preorder considers infinitely many variables, but that it does
not leave infinitely many variables unobserved. In particular, this shows that nothing changes compared to the
restricted setting if one compares unifiers on finite supersets of the set of variables occurring in the unification problem instead.

Rather surprisingly, with unification in the description logic \EL, we have found (already in the conference version~\cite{BaaderG25}) an example  
where the unification type improves from type zero to infinitary when going from the restricted to the unrestricted instantiation preorder (\autoref{EL:thm}).
In the present paper we were able to show that this is not an extraordinary property of \EL, but this also happens for the theories \AI, \FLzero, and \ACUh
(\autoref{AI:ACI:thm} and \autoref{restrictive:monoidal:thm}). However, we were also able to find a case (unification in the modal logic \ModK) where the 
unification type stays zero when going from the restricted to the unrestricted setting (\autoref{modal:k:result}).

\subsection{Discussion}

While the contributions of this work are primarily foundational, one can nevertheless ask whether the obtained results also have a practical impact. 
The answer to the question which  instantiation preorder should be used in practice mainly depends on the application that employs unification, i.e., it depends on (i)~which are the variables that are relevant in the overall procedure,  and (ii)~which of these variables must be taken into account
by the instantiation preorder between substitutions to guarantee correctness of the procedure. 
It is conceivable that there are applications that introduce more and more variables, and thus in the limit require instance relationships to hold on all these
variables, though in the usual applications of unification in term rewriting, theorem proving, and logic programming this is not the case.
This is probably the reason for the switch to the restricted instantiation preorder in the 1980s, though this was to the best of our knowledge never explicitly mentioned
in the literature.
 
We claim that, unless the application enforces comparison on all variables, the restricted instantiation preorder should be used. This is clear
for theories that are unitary or finitary w.r.t. restricted instantiation.
For example, in Knuth-Bendix completion modulo equational theories~\cite{Peterson:Stickel81,Jouannaud:Kirchner86,DBLP:journals/tcs/BachmairD89}, 
where unification is employed to test confluence by computing critical pairs, the restricted instantiation preorder is clearly sufficient, 
and thus should be used for theories that are unitary or finitary w.r.t.\ it.

But even in the case where the restricted unification type is zero and the unrestricted one is infinitary, 
it is probably not useful to employ the unrestricted preorder. While it is in principle possible to enumerate a minimal complete set of 
unifiers w.r.t.\ the unrestricted instantiation preorder (if one can find a unification algorithm that does this), this set is also a complete set of unifiers
w.r.t.\ the restricted instantiation preorder. Since this set cannot be minimal w.r.t.\ the restricted instantiation preorder, there is also a strict subset of it that is still complete.
For example, imagine a unification problem whose unifiers form an infinite decreasing chain w.r.t.\ the restricted instantiation preorder, but
where all these unifiers are incomparable w.r.t.\ the unrestricted instantiation preorder. Then the minimal complete set w.r.t.\ the unrestricted instantiation preorder
consists of all these unifiers, whereas w.r.t.\ the restricted instantiation preorder it would be sufficient to enumerate infinitely many of these unifiers, with 
arbitrarily large finite gaps in between (e.g., only every 10th element of the chain). Applied in the context of Knuth-Bendix completion, this
would mean that one might be able to find a non-confluent critical pair faster.

\subsection{Future and related work}
Except for the cases where the restricted and the unrestricted types coincide according to our \autoref{stays:same:thm} (see \autoref{stays:same:cor}) there are no unification algorithms known for the unrestricted case since research on equational unification concentrated on the restricted case.
Though we have argued above that, even in settings where the unrestricted type is better (infinitary) than the restricted one (zero), 
it makes more sense to compute a smaller, non-minimal complete set w.r.t.\ the restricted instantiation preorder than a minimal one
w.r.t.\ the unrestricted instantiation preorder, there is nevertheless the foundational challenge of how to find algorithms enumerating minimal
complete sets of unifiers in the unrestricted setting. Our proofs for the upper-bound infinitary do not directly yield practical algorithms.

Another possible topic for future research is whether one can lift the restriction to elementary unification that we have made in this paper.
Our reason for making this restriction was actually mainly for the sake of simplicity. For most of our results,
it is easy to see that they hold accordingly for unification with constants. Lower bounds transfer from elementary unification to unification
with constants since every elementary unification problem is also a problem with constants. In fact, unification with constants only means that one is allowed to use free constants in the unification problem, but one is not forced to do so.
For our upper bounds in Section~\ref{upper:bounds:sec}, it is easy to see that the results for finite and locally finite theories also hold
in the presence of free constants. For monoidal theories, adapting our proofs would, however, be more work since, e.g., the characterization of composition
of substitutions in \eqref{expre:coeff:instantiation} would need to take free constants into account. 
For general unification, results on the combination of unification algorithms \cite{DBLP:journals/jsc/Schmidt-Schauss89,DBLP:journals/jsc/BaaderS96}
imply that restricted unification type unitary or finitary for unification with linear constant restrictions, a generalization of unification with constants, 
yields restricted unification type unitary or finitary for general unification. Whether such combination results can be adapted to the unrestricted setting 
or to unification type infinitary is an interesting open problem.

Regarding results for additional equational theories´, we have stated in the conclusion of~\cite{BaaderG25} that
``it is probably not very interesting to find further equational theories where the phenomena already exhibited in this paper also occur, unless
one can prove a meta-theorem that exactly characterizes under what conditions these changes in the unification type happen for a given class 
of equational theories.'' In the present paper, we were able to show such meta-results for several classes of theories, among them a subclass of
the class of monoidal theories, which were mentioned as a candidate class in~\cite{BaaderG25}.
 
Regarding new phenomena not exemplified already in~\cite{BaaderG25}, we show here that the unification type of \ModK is zero in both the restricted and
the unrestricted setting. We still have not found an example where going from the restricted to the unrestricted setting changes the unification
type from infinitary to zero. The candidate theory \D mentioned in~\cite{BaaderG25} has turned out to be infinitary in both settings.
The cells marked with ``?'' in \autoref{results:table} describe cases that are conceivable from the order-theoretic point of view, which only takes 
into account that the unrestricted preorder is a subset of the restricted one, but for which we have not yet found examples or a theorem that would
preclude them.
To see that going from the restricted type infinitary to the unrestricted type zero is conceivable from the order-theoretic point of view,
imagine a unification problem $\Gamma$ whose unifiers are $\{\sigma_i,\theta_i\mid i\geq 0\}$ where, w.r.t.\ the unrestricted instantiation preorder, 
the substitutions $\sigma_i$ form an infinite decreasing sequence and the substitutions $\theta_i$ are incomparable with each other and the $\sigma_j$. 
Then $\Gamma$ clearly has type zero in the unrestricted setting. Now assume that, w.r.t.\ the restricted instantiation preorder, $\sigma_i$ additionally is an instance
of $\theta_i$ for $i\geq 0$, but no other instance relationships are added. Then $\Gamma$ has type infinitary in the restricted case since
$\{\theta_i\mid i\geq 0\}$ is a minimal complete set. It still needs to be seen whether there is a concrete equational theory where such a phenomenon can occur.

In addition to investigating the ``?'' cases, an interesting specific theory to be analyzed is the theory \AGh of Abelian groups with a homomorphism
since none of our general results appear to apply to it: while it is monoidal, it is not regular (and thus neither \autoref{NUOF:thm} nor 
\autoref{regular:theory:prop} applies) and it is not restrictive (this can be shown similarly as for \AG, see \autoref{AG:is:not:restrictive:lem}).
The same is true for the theories \AG and \GAUSS, but for these theories \autoref{stays:same:thm} applies because their associated semirings are
principal ideal domains. The semiring associated with \AGh is the polynomial ring $\Zbb[X]$ with one indeterminate $X$, which is a Noetherian ring, but
not a principal ideal domain. Thus, it is not clear whether \autoref{stays:same:thm} applies to this theory.
It would also be interesting to see whether the argument used to show that the restricted and the unrestricted unification type for idempotency \Idem
coincide can be extended to the class of shallow theories \cite{DBLP:journals/iandc/ComonHJ94}.

Finally, regarding related work, let us point out that
recently there has also been other work on the impact that changing the preorder on substitutions has on the unification 
type~\cite{DBLP:conf/unif/CabrerM14,DBLP:journals/corr/CabrerM15,BaLu-UNIF2015,DBLP:journals/flap/HocheSS16,DBLP:journals/mscs/SzaboS21},
but the preorders investigated there differ from the (un)restricted instantiation preorder.

%% file: lower-bounds.tex
\section{Lower bounds}\label{lbounds:sect}
In this section, we show the lower bound ``at least infinitary'' for the unrestricted unification type of the equational theories contained in the third column of \autoref{results:table}, with the exception of the theory \Assoc, whose unrestricted unification type has already been determined in the previous section (see \autoref{stays:same:cor}). Also note that the results of the first two columns of \autoref{results:table} are stated in \autoref{stays:same:cor} and the one in the last column in \autoref{modal:k:result}.

For theories whose restricted unification type is at least infinitary, the contrapositive of \autoref{typ:no:worse:thm} yields this lower bound also for the unrestricted unification type. This observation applies to the theories \AI, \D, \EL, \FLzero, and \ACUh. The theories \AI, \EL, \FLzero, and \ACUh are of type zero for elementary unification w.r.t.\ the restricted instantiation preorder. For the theory \AI of an associative and idempotent function symbol, type zero was shown in~\cite{DBLP:journals/jar/Schmidt-Schauss86} for unification with constants and in~\cite{DBLP:journals/jar/Baader86} for elementary unification.  Type zero for \FLzero and \EL were respectively proved in \cite{DBLP:journals/jsc/BaaderN01,DBLP:journals/jsc/Baader89} and \cite{DBLP:conf/rta/BaaderM09,DBLP:journals/corr/abs-1006-2289}. For \ACUh, type zero w.r.t.\ the restricted instantiation preorder was established in~\cite{DBLP:conf/ctcs/Baader89}.
For both-sided distributivity \D, type infinitary for unification with constants was shown in~\cite{Szabo83}. For elementary unification, the lower bound
``at least infinitary'' was established in Section~5.4 of~\cite{DBLP:conf/lics/KirchnerK90}.
\begin{cor}\label{EL:D:cor}
	The unrestricted unification type of \AI, \D, \EL, \FLzero, and \ACUh for elementary unification is at least infinitary.
\end{cor}

We continue by considering the theories \AC, \ACI, \ACU, and \ACUI. For elementary unification, \ACU and \ACUI are unitary w.r.t.\ the restricted instantiation preorder, whereas \AC and \ACI are finitary~\cite{DBLP:books/el/RV01/BaaderS01}. 
In contrast to the theories in \autoref{EL:D:cor}, where we can use \autoref{typ:no:worse:thm} to deduce that they are not unitary or finitary  w.r.t.\ the unrestricted instantiation 
preorder from the fact that they are of type at least infinitary for the restricted instantiation preorder, we must prove this directly for \ACUI, \ACU, \AC, and \ACI. To this end, following \cite{DBLP:conf/unif/BaaderL16}, we apply a more general result that holds for regular equational theories satisfying certain \emph{additional restrictions}.
Before formulating these restrictions and applying the mentioned general result, we need to introduce some notions and auxiliary results.

\begin{defi}
Given an equational theory $E$ and an $E$-unification problem $\Gamma$, we define $\Cmc_E(\Gamma)$ to be the set of all $E$-unifiers $\sigma$ of $\Gamma$ satisfying
	\[
	\forall y\in \VRan(\sigma). \exists x,x'\in V\text{ s.t. }
	x\ne x' \text{ and }
	%\\
	y\in \Var(\sigma(x))\cap \Var(\sigma(x')).
	\]
\end{defi}
Recall that \AC, \ACI, \ACU, and  \ACUI are regular, and thus the following result from~\cite{DBLP:conf/unif/BaaderL16} applies to them. For the sake of completeness, we reproduce here the full proof of this result, giving more detailed explanations than the formulation in~\cite{DBLP:conf/unif/BaaderL16}.
\begin{lem}[\cite{DBLP:conf/unif/BaaderL16}]
	\label{lem:uset}
	Let $E$ be a regular theory and $\Gamma=\{ s\id_E^? t \}$ an $E$-unification problem such that $\Var(s)\cap\Var(t) = \emptyset$.
	Then the set $\Cmc_E(\Gamma)$ is complete w.r.t.\ $\le_E^V$.
\end{lem}
\begin{proof}
	Let $\sigma$ be an $E$-unifier of $\Gamma$ that does not belong to $\Cmc_E(\Gamma)$, i.e., there is $y_0 \in \VRan(\sigma)$ such that there exists a unique variable $x_0$ satisfying $y_0 \in \Var(\sigma(x_0))$.
	We show how to transform $\sigma$ into an $E$-unifier $\sigma'$ of $\Gamma$ such that
	\[
	\sigma' \leq_E^V \sigma\ \ \ \mbox{and}\ \ \ \VRan(\sigma') \subset \VRan(\sigma).
	\]
	Since $\VRan(\sigma)$ is finite and $\leq_E^V$ is transitive, a finite number of consecutive applications of this transformation will yield an $E$-unifier $\sigma^*$ of $\Gamma$ such that $\sigma^* \leq_E^V \sigma$ and $\sigma^* \in \Cmc_E(\Gamma)$.
	From this, we can conclude that $\Cmc_E(\Gamma)$ is complete w.r.t.\ $\leq_E^V$.
	
	Consider the substitution $\tau :=  \{x_0 \mapsto y_0, y_0 \mapsto x_0\}$. We define the substitution $\sigma'$ as follows: 
	\[
	\sigma'(x) := \begin{cases}
		x, & \mbox{if } x = x_0,\\
		\tau(\sigma(x)),& \mbox{if } x \neq x_0.
	\end{cases}
	\]
	Let us first show that $\VRan(\sigma') \subset \VRan(\sigma)$. From the definition of $\sigma'$, it follows that
	\[
	(*)\ \ \ y \in \VRan(\sigma') \mbox{ implies } y \in \VRan(\sigma)\ \mbox{for all}\ y \in V \setminus \{x_0\}.
	\]
	To show this implication, assume that $y \in \VRan(\sigma')\setminus\{x_0\}$.
	If $y=y_0$, then $y \in \VRan(\sigma)$ is satisfied by our assumption on $y_0$.
        Thus, in the rest of the proof of the implication $(*)$ we can assume that $y \in \VRan(\sigma')\setminus\{x_0,y_0\}$.
	Since $y \in \VRan(\sigma')$, there is $x\in\Dom(\sigma')$ such that $y\in\Var(\sigma'(x))$. First note that $x\in\Dom(\sigma')$
        implies $x\neq x_0$, by the definition of $\sigma'$, and thus $\sigma'(x) = \tau(\sigma(x))$. 
	Since $y\not\in\{x_0,y_0\}$, we know that
	$y\in\Var(\sigma'(x))$ also implies $y\in\Var(\sigma(x))$. It remains to show that $x\in\Dom(\sigma)$. Otherwise, we would
	have $\sigma(x) = x = y$. But then also $\sigma'(x) = \tau(\sigma(x)) = \tau(y) = y = x$ since $y\not\in\{x_0,y_0\}$.
	This contradicts our assumption that $x\in\Dom(\sigma')$. Thus, we have proved the implication $(*)$.

	As for $x_0$, we know that $x_0$ is the unique variable such that $y_0 \in \Var(\sigma(x_0))$. Hence, since $\tau(x_0) = y_0$, we can conclude that $x_0 \not \in \Var(\tau(\sigma(x)))$ for all $x \neq x_0$. 
	This, together with $\sigma'(x_0) = x_0$, yields that $x_0 \not\in \VRan(\sigma')$. Therefore, we have shown that $\VRan(\sigma') \subseteq \VRan(\sigma)$. 
	
	Finally, to see that this is a strict inclusion, it suffices to analyze the case where $x_0 \not \in \VRan(\sigma)$. In this case, since $x_0$ is the unique variable such that $y_0 \in \Var(\sigma(x_0))$, the definition of $\sigma'$ implies that $y_0 \not\in \VRan(\sigma')$. Since $y_0 \in \VRan(\sigma)$, this yields $\VRan(\sigma') \subset \VRan(\sigma)$.
	
	Let us now show that $\sigma' \leq_E^V \sigma$. Consider the substitution $\theta := \{x_0 \mapsto \tau(\sigma(x_0))\}$. Since $x_0 \not \in \VRan(\sigma')$, 
	the following holds for all $x \in V \setminus \{x_0\}$:
	\[
	\tau(\theta(\sigma'(x))) = \tau(\sigma'(x)) = \tau(\tau(\sigma(x))) =\sigma(x).
	\]
	For $x_0$, we obtain
	\[
	\tau(\theta(\sigma'(x_0))) = \tau(\theta(x_0)) = \tau(\tau (\sigma(x_0))) = \sigma(x_0).
	\]
	This demonstrates that $\tau\theta\sigma' \id_E \sigma$, and thus, $\sigma' \leq_E^V \sigma$. 
	
	It remains to show that $\sigma'$ is an $E$-unifier of $\Gamma$. Since $E$ is a regular theory, we know that $\sigma(s)\id_E\sigma(t)$ implies $\Var(\sigma(s))=\Var(\sigma(t))$. 
	Hence, if $x_0$ occurs in one of the terms $s, t$, it must also occur in the other one, because $x_0$ is the unique variable producing $y_0$. However, this would contradict our assumption $\Var(s) \cap \Var(t) = \emptyset$. Therefore, we obtain that $x_0 \not \in \Var(\Gamma)$.
	Since $\sigma'$ is an instance of $\sigma$ on the variables $x \in V \setminus \{x_0\}$, the fact that $\sigma$ is an $E$-unifier of $\Gamma$ thus
	implies that the same is true for $\sigma'$.
\end{proof}

Together with \autoref{th:min}, this lemma yields the following result.

\begin{lem}\label{complete:lem}
	Let $E$ be a regular theory and $\Gamma=\{ s\id_E^? t \}$ an $E$-unification problem s.t.\ $\Var(s)\cap\Var(t) = \emptyset$.
	If $\Gamma$ has a minimal complete set of $E$-unifiers w.r.t.\ $\le_E^V$, then it has one that is contained in $\Cmc_E(\Gamma)$.
\end{lem}

\begin{proof}
	Since  $\Gamma$ has a minimal complete set of $E$-unifiers, the set $M$ of minimal elements of $[\Umc_E(\Gamma)]_E^V$ is complete.
	For $\Cmc_E(\Gamma)$ to be complete, it must contain for every equivalence class in $M$ at least one representative. Thus,
	by selecting for each class in $M$ one of its representatives in $\Cmc_E(\Gamma)$, we obtain a minimal complete set that is contained in $\Cmc_E(\Gamma)$.
	%\qed
\end{proof}

We are now ready to formulate the ``additional restrictions'' mentioned above.

\begin{defi}\label{nuof:def}
	Given a regular equational theory $E$, we say that the $E$-unification problem $\Gamma$ is \emph{NUOF} if the following conditions are satisfied:
	\begin{itemize}
		\item
		$\Gamma = \{ s\id_E^? t \}$ for terms $s, t$ satisfying $\Var(s) \cap\Var(t) = \emptyset$,
		\item
		there is a $\le_E^V$-minimal unifier $\sigma$ of $\Gamma$ that uses fresh variables, i.e.,
		$\VRan(\sigma)\setminus X\ne\emptyset$ where $X = \Var(s) \cup\Var(t)$, and
		\item
		this unifier $\sigma$ belongs to the set $\Cmc_E(\Gamma)$.
	\end{itemize}
\end{defi}

NUOF stands for ``not unitary or finitary,'' but we still need to show that this name is justified, i.e., that a NUOF $E$-unification problem cannot have a finite minimal complete set of unifiers w.r.t.\ $\leq_E^V$. 
Given a NUOF $E$-unification problem $\Gamma = \{ s\id_E^? t \}$, let $X = \Var(s)\cup\Var(t)$ and $\sigma$ a unifier of $\Gamma$ as required in \autoref{nuof:def}. In addition, let $x_0\in \VR(\sigma)\setminus X$. Consider the following construction of substitutions:
\begin{equation}\label{nuof:inf:substitutions}
	\sigma_z:=\sigma\tau_{x_0,z}\ \ \mbox{where $z\in {V}$ and $\tau_{x_0,z} := \{x_0\mapsto z, z\mapsto x_0\}$}.
\end{equation}
One can show that, under certain conditions on $z$, such substitutions $\sigma_z$ are $\le_E^V$-minimal unifiers that are incomparable to each other w.r.t.\ $\le_E^V$ (see \autoref{NUOF:lem} below). By \autoref{th:min}, this implies that $\Gamma$ cannot have a finite minimal complete set of unifiers w.r.t.\ $\le_E^V$
since there are infinitely many variables $z$ satisfying these conditions.

The following lemma puts together the results shown in Lemma~4 and Lemma~5 of \cite{DBLP:conf/unif/BaaderL16}.  For the sake of completeness, 
we reproduce here the full proofs of these results, giving more detailed explanations than the formulation in~\cite{DBLP:conf/unif/BaaderL16}.

\begin{lem}[\cite{DBLP:conf/unif/BaaderL16}]
	\label{NUOF:lem}
	Let $E$ be a regular equational theory $E$, $\Gamma = \{ s\id_E^? t \}$ a NUOF $E$-unification problem, $X = \Var(s)\cup\Var(t)$, $\sigma$ a unifier of $\Gamma$ as required in \autoref{nuof:def}, and let $\sigma_z$ for $z\in V$ be defined as in \eqref{nuof:inf:substitutions}.
	\begin{enumerate}
		\item
		For each $z\in V\setminus X$, $\sigma_z$ is a minimal $E$-unifier of $\Gamma$ w.r.t.\ $\le_E^V$.
		\item
		For any two distinct variables $z,z'\in V\setminus (\Dom(\sigma)\cup\VR(\sigma))$,
		$\sigma_z$ and $\sigma_{z'}$ are incomparable w.r.t.\ $\le_E^V$.
	\end{enumerate}
\end{lem}

\begin{proof}
	(1) First, note that $\sigma_z$ is an $E$-unifier of $\Gamma$ since $x_0, z \not \in X$. Let $\theta$ be an $E$-unifier of $\Gamma$ such that $\theta \leq_E^V \sigma_z$, i.e., there is a substitution $\lambda$ such that
	\[
	   \lambda(\theta(x)) \id_E \sigma_z(x)\  \mbox{for all}\ x \in V.
	\] 
	Since the composition $\tau_{x_0,z}\tau_{x_0,z}$ is the identity substitution, ``multiplying'' from the right in the previous equality with $\tau_{x_0,z}$ yields
	\[
	\lambda(\theta(\tau_{x_0,z}(x))) \id_E \sigma(x)\  \mbox{for all}\ x \in V.
	\] 
	This shows that $\theta\tau_{x_0,z} \leq_E^V \sigma$. In addition, for the same reason as for $\sigma_z$, the substitution $\theta\tau_{x_0,z}$ is also an $E$-unifier of $\Gamma$. Hence, since $\sigma$ is a $\leq_E^V$-minimal unifier, we obtain $\sigma \leq_E^V \theta\tau_{x_0,z}$, and thus
	there is a substitution $\tau$ such that
   \[
	\tau(\sigma(x)) \id_E \theta(\tau_{x_0,z}(x))\  \mbox{for all}\ x \in V.
	\]
	Consequently, ``multiplying'' again from the right with $\tau_{x_0,z}$ yields:
	\[
	\tau(\sigma_z(x)) \id_E \theta(x)\  \mbox{for all}\ x \in V.
	\] 
	Thus, we have shown that $\theta \leq_E^V \sigma_z$ implies that $\sigma_z \leq_E^V \theta$. This proves that $\sigma_z$ is a minimal $E$-unifier of $\Gamma$ w.r.t.\ $\leq_E^V$.
	
	(2) Assume that there exists a substitution $\lambda$ such that $\sigma_{z'}(x) \id_E \lambda(\sigma_z(x))$ for all $x \in V$, and let $t_0:=\sigma(x_0)$.
	By definition of $\sigma_z$ and the assumption that $z \not \in \Dom(\sigma)\cup\VRan(\sigma)$, we know that
	\[
	 \sigma_z(z) = \sigma(\tau_{x_0,z}(z)) = \sigma(x_0) = t_0\ \ \mbox{and}\ \ \sigma_{z'}(z) = \sigma(\tau_{x_0,z'}(z)) = \sigma(z) = z.
	\] 
	Consequently, $\sigma_{z'}(z) \id_E \lambda(\sigma_z(z))$ implies that $\lambda(t_0) \id_E z$. Since $E$ is a regular theory, there thus must exist a variable $y \in \Var(t_0)$ such that $\Var(\lambda(y)) = \{z\}$.
        Since $x_0 \in \VRan(\sigma)$ and $\sigma(x_0) = t_0$, we know that $y \in \VRan(\sigma)$. 
        In fact, if $x_0\not\in\Dom(\sigma)$, them $\sigma(x_0) = x_0 = t_0 = y$. Otherwise, $y \in \VRan(\sigma)$ by the definition of $\VRan$.

        Furthermore, $\sigma \in \Cmc_E(\Gamma)$ by our assumptions on $\sigma$. Hence, by definition of $\Cmc_E(\Gamma)$, there is a variable $x_1 \in V$ different from $x_0$ such that
	\[
	    \sigma(x_1)=t_1\ \ \mbox{and}\ \ y \in \Var(t_1).
	\]
	Note that $x_1 \in \Dom(\sigma)\cup\VRan(\sigma)$. In fact, if $x_1 \not\in \Dom(\sigma)$, then $\sigma(x_1) = x_1 = t_1 = y\in \VRan(\sigma)$.
        Hence, since $z, z' \not \in \Dom(\sigma)\cup\VRan(\sigma)$, we know that $z \neq x_1$ and $z' \neq x_1$. From this, and the fact that $x_0 \neq x_1$, we obtain 
	\[
	   \sigma_{z}(x_1)=\sigma(\tau_{x_0,z}(x_1)) = \sigma(x_1) \mbox{ and } \sigma_{z'}(x_1)=\sigma(\tau_{x_0,z'}(x_1)) = \sigma(x_1). 
	\]
	Consequently, $\sigma_{z'}(x_1) \id_E \lambda(\sigma_z(x_1))$ implies that $t_1 \id_E \lambda(t_1)$. Now, since $y \in \Var(t_1)$, regularity of $E$ yields that $\Var(\lambda(y)) \subseteq \Var(t_1)$.
	However, this yields a contradiction to the assumption that $z \not \in \VRan(\sigma)$. In fact, if $x_1 \in \Dom(\sigma)$, then
	$\Var(t_1) \subseteq \VRan(\sigma)$ and thus $z \in \Var(\lambda(y))$ provides us with the contradiction. If $x_1 \not\in \Dom(\sigma)$, then we
	have seen above that $t_1 = y$, and thus $t_1 \id_E \lambda(t_1)$ and $\Var(\lambda(y)) = \{z\}$ yield $z = y$. The contradiction is thus
	due to the fact that we have shown $y \in \VRan(\sigma)$.
	
	To sum up, we have shown that $\sigma_z \leq_E^V \sigma_{z'}$ does not hold. A symmetric argument yields that $\sigma_{z'} \leq_E^V \sigma_z$ also does not hold.
 \end{proof}

Since $V\setminus (X\cup \Dom(\sigma)\cup\VR(\sigma))$ is infinite, this lemma together with \autoref{th:min}
implies that $\Gamma$ cannot have a finite minimal complete set of $E$-unifiers w.r.t.\ $\le_E^V$.

\begin{thm}\label{NUOF:thm}
	If $E$ is a regular equational theory and $\Gamma$ a NUOF $E$-unification problem, then $\Gamma$ does not have a 
	finite minimal complete set of $E$-unifiers w.r.t.\ the unrestricted instantiation preorder $\le_E^V$.
\end{thm}

If one wants to apply this theorem to show that a given regular equational theory $E$ cannot have unrestricted unification type unitary or finitary,
one must find an appropriate NUOF $E$-unification problem. In particular, one must show that a substitution $\sigma$ satisfying
the properties required by \autoref{nuof:def} exists. For this purpose, the following result turns out to be useful.

\begin{cor}\label{nuof:proof:cor}
Let $E$ be a regular equational theory and $\Gamma = \{ s\id_E^? t \}$ an $E$-unification problem for terms $s, t$ satisfying $\Var(s) \cap\Var(t) = \emptyset$.
If $\Gamma$ is unitary w.r.t.\ restricted instantiation and any restricted mgu of $\Gamma$ must introduce a fresh variable, then $\Gamma$ does not have a finite 
minimal complete set of $E$-unifiers w.r.t.\ unrestricted instantiation.
\end{cor}

\begin{proof}
	If $\Gamma$ does not have a minimal complete set of $E$-unifiers w.r.t.\ unrestricted instantiation, then we are done.
        Otherwise, let $\Smc$ be such a minimal complete set. By \autoref{complete:lem}, we can assume without loss of generality 
	that $\Smc\subseteq \Cmc_E(\Gamma)$, and by \autoref{th:min} we know that the elements of $\Smc$ are $\le_E^V$-minimal. 

	Since $E$ is regular and  $\Gamma$ is of the form required by \autoref{nuof:def}, it is sufficient to show the existence of an appropriate substitution $\sigma$.
	To this purpose, let $\theta$ be an mgu of $\Gamma$ w.r.t.\ restricted instantiation. Since $\theta$ is an $E$-unifier of $\Gamma$, there is $\sigma\in \Smc$ 
	such that $\sigma\le_E^V \theta$. Since ${\le_E^V} \subseteq {\le_E^{\Var(\Gamma)}}$, this implies that $\sigma$ is also an mgu of
	$\Gamma$ w.r.t.\ the restricted instantiation preorder, and thus it must introduce a fresh variable.
        Consequently, $\sigma$ is a $\le_E^V$ minimal unifier of $\Gamma$ that uses a fresh variable and belongs to $\Cmc_E(\Gamma)$.
	An application of \autoref{NUOF:thm} thus proves the claim of the corollary.
\end{proof}

Note that, strictly speaking, the proof of this corollary does not show that the $E$-unification problem $\Gamma$ is NUOF. 
It only shows that it is of unrestricted unification type zero or NUOF. However, as will be shown in the next section, the theories \ACUI, \ACI, \ACU, and \AC 
to which we will apply this corollary in the rest of this section do not have unrestricted unification type zero.
Thus, applicability of \autoref{nuof:proof:cor} to a unification problem $\Gamma$ for these theories indeed implies that the problem $\Gamma$ is NUOF.
We start with applying the corollary to \ACUI.

\begin{cor}[\cite{DBLP:conf/unif/BaaderL16}]\label{ACUI:cor} 
	The unrestricted unification type of \ACUI for elementary unification is at least infinitary.
\end{cor}
\begin{proof}
	Since \ACUI is regular, it is sufficient to show that there is an \ACUI-unification problem $\Gamma$ that is NUOF. 
	For elementary unification, the theory \ACUI is unitary w.r.t.\ restricted instantiation, and
	according to Corollary~3.6 in \cite{DBLP:journals/tcs/BaaderB88},\footnote{Note that the authors of \cite{DBLP:journals/tcs/BaaderB88} use \ACI to denote what we call \ACUI, and also consider unification with constants. Elementary unification is a special case.}
	any most general unifier (w.r.t.\ restricted instantiation) of the \ACUI-unification problem 
	\begin{equation}\label{nuof:acui:problem}
        \Gamma := \{ f(x_1,f(x_2,x_3)) \id_\ACUI^? f(x_4, x_5)\},
        \end{equation}
	must use a fresh variable.
	Consequently, an application of \autoref{nuof:proof:cor} yields the desired result.
	%\qed
\end{proof}

In the following example, we analyze the \ACUI-unification problem \eqref{nuof:acui:problem} in more detail and explain the reason why a fresh variable is needed.

\begin{exa}\label{nuof:acui:ex}
        Consider the \ACUI-unification problem \eqref{nuof:acui:problem}, which contains five distinct variables $x_1,\ldots,x_5$. 
	In \cite{DBLP:journals/tcs/BaaderB88}, an algorithm for \ACUI-unification (with constants) was introduced, which computes a finite minimal complete set of \ACUI-unifiers w.r.t.\ the restricted instantiation preorder. By running this algorithm on $\Gamma$, we obtain the following most general unifier of $\Gamma$ w.r.t.\ $\leq_\ACUI^{\Var(\Gamma)}$:
        \begin{equation}\label{nuof:unifier}
                \begin{split}
                        \theta &:= \{x_1 \mapsto f(x_1, x_2),\ x_2 \mapsto f(x_3, x_4),\ x_3 \mapsto f(x_5,x_6), \\[.2em]
                        &  \qquad x_4 \mapsto f(x_1, f(x_3, x_5)),\ x_5 \mapsto f(x_2, f(x_4, x_6))\}.
                \end{split}
        \end{equation}
        This unifier introduces the fresh variable $x_6$, which does not occur in $\Gamma$. The question is, however, why there cannot be an mgu
	that uses less variables in its range. The reason is that each variable in $\VRan(\theta)$ stands for a minimal way of solving $\Gamma$.
	For example, the variable $x_1\in\VRan(\theta)$ occurs in the $\theta$-image of $x_1$ and $x_4$, and corresponds to the unifier
			$$
                        \theta_1 := \{x_1 \mapsto x_1,\ x_2 \mapsto \unit,\ x_3 \mapsto \unit, 
                        x_4 \mapsto x_1,\ x_5 \mapsto \unit\},
                        $$
        the variable $x_2\in\VRan(\theta)$ occurs in the $\theta$-image of $x_1$ and $x_5$, and corresponds to the unifier
                        $$
                        \theta_2 := \{x_1 \mapsto x_2,\ x_2 \mapsto \unit,\ x_3 \mapsto \unit, 
                        x_4 \mapsto \unit,\ x_5 \mapsto x_2\},
                        $$
	and similarly for the variables $x_3, \ldots, x_6$.

	To obtain, for example, $\theta_1$ as an instance of another unifier $\sigma$ via a substitution $\lambda$, 
	this unifier needs to contain a variable that occurs only in the $\sigma$-images of $x_1$ and $x_4$. In fact,
	$\lambda$ must replace every variable occurring in the $\sigma$-images of $x_2$, $x_3$, and $x_5$ with $\unit$ and every variable
	occurring in the $\sigma$-images of $x_1$ and $x_4$ with $\unit$ or $x_1$. To obtain $\lambda\sigma(x_1) \id_\ACUI x_1 \id_\ACUI \lambda\sigma(x_4)$,
        $\sigma(x_1)$ must contain a variable $y$ such that $\lambda(y)\id_\ACUI x_1$. Since $\sigma$ is a unifier, this variable $y$ 
	must also occur in $\sigma(f(x_4, x_5))$. However, since all variables in the $\sigma$-image of $x_5$ are replaced by $\lambda$ with $\unit$,
	$y$ must occur in $\sigma(x_4)$. Finally, $y$ cannot occur in the $\sigma$-images of $x_2$ and $x_3$, since again such variables
	need to be replaced by $\lambda$ with $\unit$. This shows that indeed $y$ is a variable that occurs only in the $\sigma$-images of $x_1$ and~$x_4$.

        Overall, there are $3\cdot 2 = 6$ ways of choosing one variable on the left-hand side and one on the right-hand side, which explains why 
	an mgu needs six variables in its range to be able to have all the substitutions $\theta_1,\ldots,\theta_6$ as instances.
\end{exa}

Next, we consider the theory \ACI. We show that the unification problem $\Gamma$ in \eqref{nuof:acui:problem} is also a NUOF \ACI-unification problem.

\begin{cor}\label{ACI:cor}
	The unrestricted unification type of \ACI for elementary unification is at least infinitary.
\end{cor}
\begin{proof}	
	Consider the unification problem $\Gamma$ in \eqref{nuof:acui:problem}.
	If $\Gamma$ does not have a minimal complete set of \ACI-unifiers w.r.t.\ unrestricted instantiation, then we are done. Thus, assume that $\Gamma$ has a minimal complete set $\Smc$ of unifiers w.r.t.\ unrestricted instantiation. By \autoref{complete:lem}, we can assume without loss of generality that $\Smc\subseteq \Cmc_\ACI(\Gamma)$, and by \autoref{th:min} we know that the elements of $\Smc$ are $\le_\ACI^V$-minimal. We show that $\Smc$ contains a substitution $\sigma$ that uses a fresh variable, and consequently, $\Gamma$ is also a NUOF \ACI-unification problem.
	
	Consider the substitution $\theta$ defined in \eqref{nuof:unifier}. We know that $\theta$ is an \ACUI-mgu of $\Gamma$ (w.r.t.\ restricted instantiation) and it does not map any variable to the unit $\unit$. Consequently, $\theta$ is also an \ACI-unifier of $\Gamma$.
	Hence, there is $\sigma \in \Smc$ such that $\sigma \leq_\ACI^V \theta$. Since ${\leq_\ACI^V} \subseteq  {\leq_\ACI^{\Var(\Gamma)}}$, this implies that $\sigma \leq_\ACI^{\Var(\Gamma)} \theta$.
	In addition, $\sigma$ is obviously also an \ACUI-unifier of $\Gamma$. 
	%	%
	Since ${\le_\ACI^{\Var(\Gamma)}} \subseteq {\le_\ACUI^{\Var(\Gamma)}}$ and $\theta$ is a most general unifier of $\Gamma$ w.r.t.\ $\leq_\ACUI^{\Var(\Gamma)}$, we obtain
	\[
	\sigma \le_\ACUI^{\Var(\Gamma)}\theta \ \ \  \mbox{and}\ \ \  \theta\le_\ACUI^{\Var(\Gamma)} \sigma.
	\]  
	This means that $\sigma$ is an \ACUI-mgu of $\Gamma$ w.r.t.\ the restricted instantiation preorder. According to Corollary~3.6 in \cite{DBLP:journals/tcs/BaaderB88}, any most general unifier (w.r.t.\ restricted instantiation) of $\Gamma$ must use a fresh variable. Thus, $\sigma$ introduces a fresh variable.
		 
	Consequently, we have shown that $\Gamma$ is a NUOF \ACI-unification problem. Thus, since \ACI is regular, we can apply \autoref{NUOF:thm} to conclude the proof of the corollary.
\end{proof}

Interestingly, elementary unification in \ACU and in \AC can be treated by almost the same proof as the one for \autoref{ACI:cor}.
To see this, first note that the unifier $\theta$ introduced in \eqref{nuof:unifier} is also an \ACU- and an \AC-unifier of the
unification problem $\Gamma$ defined in \eqref{nuof:acui:problem}. Second, observe that both ${\le_\ACU^{\Var(\Gamma)}} \subseteq {\le_\ACUI^{\Var(\Gamma)}}$
and ${\le_\AC^{\Var(\Gamma)}} \subseteq {\le_\ACUI^{\Var(\Gamma)}}$ holds. Given these facts, the proof of \autoref{ACI:cor} also goes through 
if we replace \ACI with \ACU or with \AC.

\begin{cor}\label{ACU:cor}
	The unrestricted unification types of \ACU and of \AC for elementary unification are at least infinitary.
\end{cor}

%% file: upper-bounds.tex
\section{Upper bounds}\label{upper:bounds:sec}
This section is devoted to obtaining matching ``at most infinitary'' upper bounds for the ``at least infinitary'' lower bounds established in the previous section,
by proving that none of these theories can have type zero. In this way, we determine the exact unrestricted unification type (infinitary) of these equational theories.

To demonstrate that the theories in question cannot have unrestricted unification type zero, we prove that all these theories are Noetherian. 

\begin{defi}
We call an equational theory $E$ \emph{Noetherian} w.r.t.\ the unrestricted instantiation preorder if
\begin{itemize}
        \item
        every strictly descending chain of substitutions $\sigma_1 >_E^V \sigma_2>_E^V \sigma_3 >_E^V \cdots $ is finite.
\end{itemize}
\end{defi}
The notion of a Noetherian equational theory was introduced in~\cite{BurckertHS89} for the restricted instantiation preorder where non-existence
of an infinite strictly descending chain of substitutions was required for $>_E^X$ for all finite subsets $X$ of $V$. 
The authors of~\cite{BurckertHS89} show that Noetherian in their sense precludes restricted unification type zero.
We will see that there are theories that are not Noetherian in the sense of~\cite{BurckertHS89}, but are Noetherian w.r.t.\ the unrestricted instantiation preorder.
In fact, the theories of restricted type zero considered in this section (\EL, \FLzero, \ACUh, \AI) cannot be
Noetherian in the restricted setting, but we will show that these theories are Noetherian in the unrestricted setting.\footnote{%
See Section~3 in~\cite{BaaderG25} for a concrete example of an \EL unification problem where the set of unifiers contains an infinite strictly descending chain
w.r.t.\ the restricted instantiation preorder, but not w.r.t.\ the unrestricted one.
}
In the following, when we say Noetherian we will mean Noetherian w.r.t.\ the unrestricted instantiation preorder.

\begin{prop}\label{noeth:prop}
If $E$ is a Noetherian theory, then it cannot have unrestricted unification type zero.
\end{prop}

\begin{proof}
By \autoref{th:min} it is sufficient to show that the set of $\preceq_E^V$-minimal elements in the set of all unifiers is complete, i.e.,
that every unifier $\theta$ is above a minimal unifier. If $\theta$ itself is minimal, then we are done. Otherwise, there is a unifier
$\sigma_1$ such that $\theta>_E^V \sigma_1$. Again, if $\sigma_1$ is minimal, we are done. Otherwise, there is a unifier $\sigma_2$ such that 
$\sigma_1 >_E^V \sigma_2$. Continuing like this, we create a strictly descending sequence
$\theta>_E^V \sigma_1>_E^V \sigma_2>_E^V \sigma_3 >_E^V \cdots$. Since such a sequence must be finite, there is a minimal unifier $\sigma_n$
such that $\theta>_E^V \sigma_1>_E^V \sigma_2>_E^V \sigma_3 >_E^V \cdots >_E^V \sigma_n$. This shows that any unifier is above a minimal one.
\end{proof}

Our proof that the equational theories under consideration here are Noetherian proceeds in two steps.
First, we show in \autoref{regular:theory:prop} that regular theories satisfy a property that allows us to reduce
the search for a minimal unifier $\sigma$ below a given unifier $\theta$ to substitutions that do not use more variables than $\theta$. 
Then, we prove that finite theories are Noetherian w.r.t.\ the unrestricted instantiation preorder. This allows us to establish the 
``at most infinitary'' upper bounds for the theories \AC and \D.
Next, we consider regular, locally finite theories, such as \AI, \ACI, and \ACUI, and prove that they are Noetherian. The theory \EL is regular,
but not locally finite. However, a proof similar to the one employed for locally finite theories can be used to establish that
\EL is Noetherian.
Finally we turn our attention to the class of monoidal theories. We define the subclass of ``restrictive'' monoidal theories and prove that the regular theories of this subclass are also Noetherian. Since \ACU, \ACUh, and \FLzero are all members of this subclass, this implies that their unrestricted unification type cannot be zero.

\subsection{A crucial proposition}\label{proof:regular:theories:prop:sec}

The following proposition is a vital part of our proofs of the ``at most infinitary'' upper bounds in the subsequent subsections. 
It shows that any substitution $\sigma$ more general than a given substitution $\theta$ is equivalent to one not using more variables than $\theta$, i.e.,
both the domain and the variable range of the equivalent substitution $\sigma'$ is contained in $\Dom(\theta)\cup\VRan(\theta)$. 
This strongly restricts the candidates for being substitutions more general than $\theta$ since it shows that it is sufficient to consider
only substitutions using variables from the finite set of variables $\Dom(\theta)\cup\VRan(\theta)$ determined by $\theta$.
As mentioned above, this is an important step towards proving that the theories under consideration in this section are Noetherian.
The second step that then actually establishes this property depends on the specific equational theory $E$, and we will show in the subsequent subsections how this step
can be made for three different classes of theories.
\begin{prop}\label{regular:theory:prop}
        Let $E$ be a regular theory and let $\sigma, \theta$ be substitutions such that $\sigma\leq^V_E \theta$.
        Then there is a substitution $\sigma'$ such that $\sigma \sim_E^V \sigma'$ and $\Dom(\sigma') \cup \VRan(\sigma')\subseteq\Dom(\theta)\cup\VRan(\theta)$.
\end{prop}
In \cite{BaaderG25}, we have proved the result stated in \autoref{regular:theory:prop} for the specific regular theories \ACUI, \ACU, \AC, and \EL. 
Our proofs consisted of two main steps. First, we were able to show that a substitution $\sigma'$ always exists such that\footnote{% 
See, for example, Lemma~18 in \cite{BaaderG25} for \EL.}
\begin{equation}\label{previous:crucial:lem}
	\sigma \sim^V_E \sigma'\ \ \mbox{ and }\ \ \Dom(\sigma') \subseteq \Dom(\theta).
\end{equation}
From this, it was relatively easy to show that $\VRan(\sigma') \subseteq \Dom(\theta) \cup \VRan(\theta)$ also holds.\footnote{%
See, for example, Lemma 19 in \cite{BaaderG25} for \EL.} 
However, if we consider an arbitrary regular theory, a substitution $\sigma'$ as required in \eqref{previous:crucial:lem} need not always exist. 
For this reason, showing the result in \autoref{regular:theory:prop} (for all regular theories) requires a new proof. Before providing this new proof, we give an example of a regular theory not satisfying \eqref{previous:crucial:lem}.
\begin{exa}\label{previous:crucial:lem:not:true:ex}
	Consider a signature $\Sigma$ consisting of two unary function symbols $f$ and $g$, and let $E$ be the following equational theory:
	\[
	E:=\{f(g(x)) \id x\}.
	\]
	Obviously, $E$ is a regular theory. In addition, it is not hard to show that for every term $t\in T(\Sigma,V)$ and variable $y \in V$, if $t \id_E y$ then $t=y$ or $t$ is of the form $f(\cdots (g(y))\cdots)$.
	
	Now, let $\sigma$ and $\theta$ be the following substitutions:
	\[
	    \sigma := \{x \mapsto y,\  y \mapsto f(y)\}\ \ \ \mbox{and}\ \ \theta := \{x \mapsto g(y)\},
	\]
	which satisfy $\Dom(\sigma) \not\subseteq \Dom(\theta)$. It is easy to verify that $\{y\mapsto g(y)\}\sigma \id_E \theta$, and thus, $\sigma\leq_E^V \theta$.
	Let us then assume that there is a substitution $\sigma'$ satisfying \eqref{previous:crucial:lem}. This means that $\sigma'$ must satisfy the following:
	\begin{itemize}
		\item 
		$y \not \in \Dom(\sigma')$, and
		\item 
		there are substitutions $\lambda$ and $\tau$ such that $\lambda\sigma \id_E \sigma'$ and $\tau\sigma' \id_E \sigma$.
	\end{itemize}
	From this, we obtain that
	\[
	    \lambda(\sigma(x)) = \lambda(y) \id_E \sigma'(x) \ \ \mbox{and} \ \ \lambda(\sigma(y)) = \lambda(f(y)) \id_E \sigma'(y)=y. 
	\]
	Since $\lambda(f(y)) \id_E y$, we know that $\lambda(f(y))$ is of the form $f(\cdots (g(y))\cdots)$. Hence, the term $\lambda(y)$ must be of the form $\Xi(\cdots(g(y))\cdots)$, where $\Xi \in \{f, g\}$ may also coincide with the inner $g$. 
	Now, by using $\tau\sigma' \id_E \sigma$, we can further derive that
	\[
	\tau(\sigma'(x)) \id_E \tau(\lambda(y)) = \tau(\Xi(\cdots(g(y))\cdots)) \id_E \sigma(x) = y\ \ \mbox{and} \ \ \tau(\sigma'(y)) = \tau(y) \id_E f(y).
	\]
	By combining these equivalences, we obtain that
	\[
	  \tau(\Xi(\cdots(g(y))\cdots)) = \Xi(\cdots(g(\tau(y)))\cdots)) \id_E \Xi(\cdots(g(f(y)))\cdots)) \id_E y.
	\]
	However, as pointed out above, a term $t$ of the form $\Xi(\cdots(g(f(y)))\cdots))$ does not satisfy $t\id_E y$. Thus, this contradicts the existence of a substitution $\sigma'$ satisfying \eqref{previous:crucial:lem}. 
\end{exa}

Let us continue with the proof of \autoref{regular:theory:prop}. The first step is to show that regular theories satisfy the property stated in the following lemma. 
Intuitively, this property plays a r\^ole similar to the one of \eqref{previous:crucial:lem} in the proofs in~\cite{BaaderG25}.

\begin{lem}\label{reg:theories:crucial:lem}
	Let $E$ be a regular theory an let $\sigma, \theta$ be substitutions such that $\sigma\leq_E^V \theta$. Then there is a substitution $\sigma'$ such that
	\[
	\sigma \sim_E^V \sigma'\ \ \mbox{ and }\ \ \Var(\sigma'(x)) =\{x\} \mbox{ for all } x \in \Dom(\sigma')\setminus \Dom(\theta).
	\]
\end{lem}

\begin{proof}
	First, note that we can assume without loss of generality that $\VRan(\sigma)\cap\Dom(\theta) = \emptyset$. Otherwise, we can apply a permutation
	to $\sigma$ that renames the variables in $\VRan(\sigma)$ appropriately, which yields a substitution that is $\sim_E^V$-equivalent to $\sigma$
	and satisfies the required disjointness condition
	(see the proof of \autoref{types:same:lem} for how such a permutation can be obtained).
	
	Since $\sigma\leq_E^V \theta$, we know that there is a substitution $\lambda$ such that $\lambda\sigma\id_E\theta$.
	Consider the set $\Xmc$ of all variables $x$ such that $x\in\Dom(\sigma)\setminus \Dom(\theta)$. Then 
	$\lambda(\sigma(x))\id_E\theta(x) = x$ holds for all $x\in\Xmc$. 
	Hence, an application of \autoref{regular:implies:all:regular:lem} yields $\Var(\lambda(\sigma(x))) = \{x\}$, and thus %
	\[
	\Var(\sigma(x)) = \{z_1^x,\ldots, z_{n_x}^x\}\ \ \mbox{and}\ \ \Var(\lambda(z_1^x)) \cup \cdots\cup  \Var(\lambda(z_{n_x}^x)) = \{x\}.
	\]
	Let us assume without loss of generality that $\Var(\lambda(z_1^x)) = \{x\}$. This implies that $z_1^x\neq z_1^y$ for different elements $x, y$ of \Xmc.
	
	Since $\VRan(\sigma)\cap\Dom(\theta) = \emptyset$, we know that $z_i^x\not\in\Dom(\theta)$ for all $x\in\Xmc$ and $1\leq i\leq n_x$. We claim that $z_i^x\in\Dom(\sigma)$ also holds, and thus $z_i^x\in\Xmc$. 
	Otherwise, $\lambda(\sigma(z_i^x)) = \lambda(z_i^x)$ and $\Var(\lambda\sigma(z_i^x)) \subseteq \{x\}$. However, since $E$ is regular and $\lambda(\sigma(z_i^x))\id_E\theta(z_i^x) = z_i^x$, an application of \autoref{regular:implies:all:regular:lem} yields that $\Var(\lambda(\sigma(z_i^x)))=\{z_i^x\}$.
	For this not to yield a contradiction, $z_i^x = x$ must hold. But then we also have $z_i^x\in\Dom(\sigma)$ since $x\in\Dom(\sigma)$.
	We have thus shown that all the variables $z_i^x$ for $x\in \Xmc$ also belong to $\Xmc = \Dom(\sigma)\setminus\Dom(\theta)$.
	Since $z_1^x\neq z_1^y$ for different elements $x, y$ of \Xmc, this implies that the $z$-variables with index $1$ already ``use up''
	all of \Xmc, and thus $n_x=1$ holds for all $x\in \Xmc$.
	
	Consider the substitution $\tau$ defined as:
	\[
	\tau(x) := \begin{cases}
		z_1^x & \mbox{if }x \in \Dom(\sigma)\setminus\Dom(\theta),\\[.3em]
		x & \mbox{otherwise}.
	\end{cases}
	\]
	According to what we have shown above, this substitution $\tau$ is an $\Xmc$-renaming for $\Xmc = \Dom(\sigma)\setminus\Dom(\theta)$, 
        where according to Definition~2.11 in \cite{DBLP:journals/jsc/Eder85}
	a substitution $\tau$ is a $W$-renaming for a set $W\subseteq V$ if $\tau(x)$ is a variable for all $x\in W$ and $\tau$ is injective on $W$.
	Since $\tau$ is the identity on $V\setminus (\Dom(\sigma)\setminus\Dom(\theta))$ and 
	elements of $\Dom(\sigma)\setminus\Dom(\theta)$ are mapped to $\Dom(\sigma)\setminus\Dom(\theta)$ by $\tau$, 
	the substitution $\tau$ is also a $W$-renaming for $W = V$.
	By Lemma~2.12 in~\cite{DBLP:journals/jsc/Eder85}, there is a permutation $\pi$ that coincides with $\tau$ on $V$.
	Let $\sigma' := \pi^{-1}\sigma$. Then $\sigma'\sim_\emptyset^V \sigma$, and thus also $\sigma'\sim_E^V \sigma$. In addition, it is easy to see that $\Dom(\sigma') \subseteq \Dom(\sigma)$. Hence, for all $x \in \Dom(\sigma') \setminus \Dom(\theta)$ we also have $x \in \Dom(\sigma) \setminus \Dom(\theta)$.
	From this, it follows that
	\[
	\Var(\sigma(x)) = \{z_1^x\}\ \ \mbox{and}\ \ \pi(x) = \tau(x) = z_1^x.
	\]
	Thus, since $\sigma' (x) = \pi^{-1}(\sigma(x))$, we obtain $\Var(\sigma'(x)) = \{x\}$.
\end{proof}

 This lemma together with the next result shows that a substitution $\sigma'$ as required in \autoref{regular:theory:prop} always exists.

\begin{lem}\label{regular:theory:main:lem}
	Let $E$ be a regular theory and let $\sigma, \theta$ be substitutions such that 
	\[
	\sigma\leq^V_E \theta\ \ \mbox{and}\ \ \Var(\sigma(x)) =\{x\} \mbox{ for all } x \in \Dom(\sigma)\setminus \Dom(\theta).
	\] 
	Then there is a substitution $\sigma'$ such that $\sigma \sim_E^V \sigma'$ and $\Dom(\sigma') \cup \VRan(\sigma')\subseteq\Dom(\theta)\cup\VRan(\theta)$.
\end{lem}

\begin{proof}
	Let $\lambda$ be such that $\lambda\sigma\id_E\theta$. Consider the set $\Xmc$ of all variables $x$ such that $x \in \Dom(\sigma) \setminus (\Dom(\theta)\cup\VRan(\theta))$. Then, $\lambda(\sigma(x)) \id_E \theta(x) = x$ holds for all $x \in \Xmc$. Hence, an application of \autoref{regular:implies:all:regular:lem} yields that $\Var(\lambda(\sigma(x))) = \{x\}$.	
	Due to $\Var(\sigma(x)) =\{x\}$ we have that $\Var(\lambda(\sigma(x)))=\Var(\lambda(x))$. It follows that $\Var(\lambda(x)) = \{x\}$ for all $x \in \Xmc$. From this, we can conclude that the following holds (for all $x\in \Xmc$):
	\begin{equation}\label{x:not:in:vars:y}
		x \not \in \Var(\sigma(y))\ \ \mbox{for all variables } y \neq x.
	\end{equation}
	In fact, if $x \in \Var(\sigma(y))$, then $\Var(\lambda(x)) = \{x\}$ implies that $x \in 
	\Var(\lambda(\sigma(y)))$. Hence, since $\lambda(\sigma(y)) \id_E \theta (y)$ and $E$ is regular, \autoref{regular:implies:all:regular:lem} tells us that $x \in \Var(\theta(y))$. If $y\in\Dom(\theta)$, then this contradicts the fact that $x \not\in\VRan(\theta)$; otherwise, it contradicts the fact that $y \neq x$.
	
	Let us define $\sigma'$ as the following substitution:
	\[
	\sigma'(x) := \begin{cases}
		x & \mbox{if } x \in \Xmc,\\[.3em]
		\sigma(x) & \mbox{otherwise}. 
	\end{cases}
	\]
	First, we show that $\sigma \sim_E^V \sigma'$. It suffices to find two substitutions $\tau_1$ and $\tau_2$ such that $\tau_1\sigma\id_E\sigma'$ and $\tau_2\sigma' \id_E \sigma$. Define $\tau_1$ and $\tau_2$ as
	\begin{align*}
		\tau_1(x) :=
		\begin{cases}
			\lambda(x) & \mbox{if } x\in \Xmc\\[.3em]
			x & \mbox{otherwise}. 
		\end{cases}
		&&
		\tau_2(x):=\begin{cases}
			\sigma(x) & \mbox{if } x\in \Xmc\\[.3em]
			x & \mbox{otherwise}. 
		\end{cases}
	\end{align*}
	To show that these substitutions satisfy the required properties, we make the following case distinction:
	\begin{itemize}
		\item
		$x \in \Xmc$. 
		We know that $\Var(\sigma(x)) = \{x\}$ and $\tau_1(x)=\lambda(x)$, and thus $\tau_1(\sigma(x)) = \lambda(\sigma(x))$. By the definition of $\sigma'$, we consequently obtain
		\[
		\tau_1(\sigma(x)) = \lambda(\sigma(x)) \id_E \theta(x) = x = \sigma'(x).
		\]
		In addition, the definitions of $\tau_2$ and $\sigma'$ yield $\tau_2(\sigma'(x)) = \tau_2(x) = \sigma(x)$.
		\item
		$y \not\in \Xmc$. Then \eqref{x:not:in:vars:y} yields $x \not \in \Var(\sigma(y))$ for all $x \in \Xmc$. Hence, since $\sigma(y)=\sigma'(y)$ and $\tau_1$ and $\tau_2$ behave as the identity on $V \setminus \Xmc$, it follows that
		\[
		\tau_1(\sigma(y)) = \sigma(y) = \sigma'(y) = \tau_2(\sigma'(y)).
		\]     
	\end{itemize}
	This proves that $\tau_1(\sigma(x))\id_E\sigma'(x)$ and $\tau_2(\sigma'(x)) \id_E \sigma(x)$ for all $x \in V$. Thus, we have shown that $\tau_1\sigma\id_E\sigma'$ and $\tau_2\sigma' \id_E \sigma$, as required.
	
	Second, we prove that $\Dom(\sigma') \cup \VRan(\sigma') \subseteq \Dom(\theta) \cup \VRan(\theta)$.
	The definition of $\sigma'$ clearly yields $\Dom(\sigma') = \Dom(\sigma) \setminus \Xmc$, and thus % 
        (since $\Xmc = \Dom(\sigma) \setminus (\Dom(\theta)\cup\VRan(\theta))$)
	\[
	\Dom(\sigma') \subseteq \Dom(\theta) \cup \VRan(\theta).
	\]
	Moreover, for all $z \in \VRan(\sigma')$ there is $y \not \in \Xmc$ such that $z$ occurs in $\sigma(y)$. Then, \eqref{x:not:in:vars:y} implies that $z \not \in \Xmc$. This shows that
	\[
	\VRan(\sigma') \subseteq \VRan(\sigma) \ \ \mbox{and}\ \ \VRan(\sigma') \cap \Xmc = \emptyset.
	\]
	Assume there is $z \in \VRan(\sigma')$, but $z \not\in \Dom(\theta) \cup \VRan(\theta)$. Then, $\Dom(\sigma') \subseteq \Dom(\theta) \cup \VRan(\theta)$ implies that $z \not \in \Dom(\sigma')$. Hence, since $\Dom(\sigma') = \Dom(\sigma) \setminus \Xmc$ and $z\not\in \Xmc$, we have that $z \not \in \Dom(\sigma)$. 
	This means that $\lambda(z) = \lambda(\sigma(z)) \id_E \theta(z) = z$. Hence, an application of \autoref{regular:implies:all:regular:lem} yields that $\lambda(z)$ contains $z$. Since $z\in \VRan(\sigma)$, there is a variable $x\in\Dom(\sigma)$ such that $\sigma(x)$ contains $z$ and $x \neq z$. This implies that $\lambda(\sigma(x))$ contains $z$.
	Therefore, since $\lambda(\sigma(x)) \id_E \theta(x)$ and $E$ is regular, it follows from \autoref{regular:implies:all:regular:lem} that $\theta(x)$ contains $z$. Then, $x \neq z$ implies that $x \in \Dom(\theta)$. Hence, $z \in \Var(\theta(x))$ implies that $z \in \VRan(\theta)$, which contradicts our assumptions on $z$.
	Thus, we have shown that $\VRan(\sigma') \subseteq \Dom(\theta) \cup \VRan(\theta)$, which concludes the proof of this lemma.
\end{proof}

Thus, given two substitutions $\sigma$ and $\theta$ such that $\sigma\leq^V_E \theta$, the consecutive applications of \autoref{reg:theories:crucial:lem} and \autoref{regular:theory:main:lem} yield a substitution $\sigma'$ satisfying the properties required in \autoref{regular:theory:prop}.
This concludes the proof of our crucial proposition.

\subsection{Finite theories}

In \cite{BurckertHS89}, it was shown that finite equational theories cannot have unification type zero w.r.t.\ the restricted instantiation preorder. 
Our proof that the same holds for the unrestricted case
is similar to the one employed in \cite{BurckertHS89} for the restricted case, but it requires several adjustments to deal with the unrestricted 
instantiation preorder $\leq_E^V$. We will make use of the following simple observation.
\begin{rem}\label{finite:set:rem}
	Let $t \in T(\Sigma,V)$ and $X\subseteq V$ be a finite set of variables. Then, there are only finitely many terms $s \in T(\Sigma,X)$ such that $\theta(s) \id_\emptyset t$ for some substitution $\theta$.
\end{rem}

The following lemma shows that finite equational theories are  Noetherian w.r.t.\ $\leq_E^V$. It is the analog, for the unrestricted case, of item i) in Lemma 3.3.8 of \cite{BurckertHS89}.%

\begin{lem}\label{finite:implies:noetherian:lem}
	If a theory is finite, then it is Noetherian w.r.t.\ the unrestricted instantiation preorder.
\end{lem}
\begin{proof}
	Let $E$ be a finite equational theory. By \autoref{finite:is:regular:cfree:lem}, we know that $E$ is also a regular theory. 
	Suppose that $E$ is not Noetherian, i.e., there exists a strictly decreasing chain $\theta_1 >_E^V \theta_2>_E^V \theta_3 >_E^V \cdots$ of substitutions. Since $\leq_E^V$ is a transitive relation, we have that $\theta_i \leq^V_E \theta_1$ for all $i \geq 2$. Hence, since $E$ is regular, we can apply \autoref{regular:theory:prop} to obtain substitutions $\sigma_2, \sigma_3,\ldots$ such that the following holds for all $i \geq 2$:
	\[
	\theta_i \sim_E^V \sigma_i\ \ \mbox{and}\ \ \Dom(\sigma_i) \cup \VRan(\sigma_i) \subseteq \Dom(\theta_1) \cup \VRan(\theta_1).
	\]
		Based on this, we prove that there are indices $i, j$ such that $i \neq j$ and $\sigma_i(x) = \sigma_j(x)$ for all $x \in V$. Note that this implies that $\theta_i \sim_E^V \theta_j$. Hence, since $i \neq j$, this would contradict the assumption that the above chain of $\theta_i$'s is strictly decreasing.
	
	From $\Dom(\sigma_i) \subseteq \Dom(\theta_1) \cup \VRan(\theta_1)$, we obtain that all substitutions $\sigma_i$ coincide on $V \setminus (\Dom(\theta_1) \cup \VRan(\theta_1))$.
	Now, let $x \in \Dom(\theta_1) \cup \VRan(\theta_1)$. We make the following observations:
	\begin{itemize}
		\item
		From $\theta_i \leq^V_E \theta_1$ and $\theta_i \sim_E^V \sigma_i$, it follows that there is a substitution $\lambda_i$ such that $\lambda_i \sigma_i \id_E^V  \theta_1 $. This means that, for all $i \geq 2$, the following holds:
		\[
		\lambda_i (\sigma_i(x)) \id_E \theta_1(x).
		\]
		\item
		Consider the following set of terms:
		\[
		S := \{\lambda_i(\sigma_i(x)) \mid i \geq 2\}.
		\]
		Since $E$ is a finite equational theory, the $\id_E$-equivalence class of $\theta_1(x)$ is finite. Hence, the set $S$ is also finite.
		\item
		Since $\VRan(\sigma_i) \subseteq \Dom(\theta_1)\cup\VRan(\theta_1)$ for all $i \geq 2$, we know that $\sigma_i(x) \in T(\Sigma,X)$ for all $i \geq 2$, where $X=\Dom(\theta_1)\cup\VRan(\theta_1)$ is a finite set of variables. Hence, since $\lambda_i(\sigma_i(x)) \in S$ and $S$ is a finite set, it follows from \autoref{finite:set:rem} that $\{\sigma_i(x) \mid i \geq 2\}$ is a finite set.
	\end{itemize}
	Thus, since all substitutions $\sigma_i$ coincide on $V \setminus (\Dom(\theta_1) \cup \VRan(\theta_1))$ and $\Dom(\theta_1) \cup \VRan(\theta_1)$ is a finite set, the third observation implies that $\sigma_i = \sigma_j$ for some $i\neq j$. This, as argued above, contradicts the existence of an infinite strictly decreasing chain.
	Thus, we have proved our claim.
\end{proof}

By \autoref{noeth:prop}, this lemma implies that the unrestricted unification type of a finite equational theory cannot be zero. This was already established for the finite theories $\emptyset$, $\Commu$, and $\Assoc$ in Section~\ref{general:results:sect}, where \autoref{stays:same:cor} also yields their exact unrestricted unification type. 
In \autoref{finite:theories:prop}, we also identified \AC and \D as finite theories. We have shown, in \autoref{ACU:cor} for \AC and in \autoref{EL:D:cor} for \D, that the unrestricted unification type of these two theories for elementary unification is at least infinitary.

\begin{thm}\label{Distr:AC:thm}
	The unrestricted unification type of \AC and \D for elementary unification is infinitary.
\end{thm}

\subsection{The unrestricted unification types of locally finite theories and of \EL}\label{EL:sect}

In universal algebra~\cite{DBLP:books/daglib/0067494}, an equational theory $E$ is called locally finite if all finitely generated algebras that are models of $E$ are finite. In terms of $\id_E$-equivalence classes, this condition can be formulated as follows.

\begin{defi}
Let $E$ be an equational theory over the signature $\Sigma$, and let $[t]_{\id_E}$ denote the $\id_E$-equivalence class of the term $t\in T(\Sigma,V)$. 
Then $E$ is called \emph{locally finite} if the set $\{[t]_{\id_E} \mid t\in T(\Sigma,X)\}$ is finite for all finite subsets $X$ of $V$.
\end{defi}

Note that, in a certain sense, finite and locally finite theories are opposites of each other. Finite means that each equivalence class $[t]_{\id_E}$ consists of finitely many terms, and
thus there must be many of them. Locally finite means that there are only finitely many classes, and thus the classes must be large. In particular,
if $\Sigma$ contains a function symbol of arity $> 0$, then $T(\Sigma,X)$ is infinite for all finite sets $\emptyset\neq X\subseteq V$. Thus, if $E$ is locally finite,
then at least one equivalence class of terms in $T(\Sigma,X)$ must be infinite. Conversely, if $E$ is finite, then there must be infinitely many
equivalence classes. Thus, finite and locally finite are mutually exclusive.

Nevertheless, like regular finite theories, the locally finite theories that are regular are Noetherian.

\begin{prop}\label{locally:finite:is:noeth:prop}
If $E$ is a regular, locally finite theory over the signature $\Sigma$, then $E$ is Noetherian. 
\end{prop}

\begin{proof}
Consider the set $\{\sigma \mid \sigma\leq^V_E \theta\}$ for a given substitution $\theta$. Since $E$ is regular, an application of \autoref{regular:theory:prop} shows that, up to $\sim_E$-equivalence, one can restrict the attention to substitutions $\sigma$ such that $\Dom(\sigma) \cup \VRan(\sigma) \subseteq \Dom(\theta) \cup \VRan(\theta)$.
Let $S$ be the subset of such substitutions and $X=\Dom(\theta) \cup \VRan(\theta)$. Then $\sigma(x) \in T(\Sigma,X)$ for all $\sigma \in S$ and $x \in \Dom(\sigma)$. Therefore, since $X$ is a finite set and $\Dom(\sigma)$ is contained in the finite set $\Dom(\theta) \cup \VRan(\theta)$, the fact that $E$ is locally finite implies that $S$ is finite up to $\id_E$-equivalence. Overall, we thus can conclude that $\{\sigma \mid \sigma\leq^V_E \theta\}$ must be finite up to $\sim_E$-equivalence.
Thus, there cannot be an infinite, strictly decreasing chain w.r.t.\ $\leq^V_E$ issuing from $\theta$. Since $\theta$ was an arbitrary substitution, this
proves that $E$ is Noetherian.
\end{proof}

Among the theories introduced in Section~\ref{classes:equational:theories:sec}, the following theories are locally finite.

\begin{prop}
The theories \AI, \ACI, and \ACUI are locally finite.
\end{prop}

\begin{proof}
For \ACI and \ACUI, this is an easy consequence of the fact that equivalence modulo these theories can be characterized by the set of variables the terms contain,
i.e., $s$ and $t$ are equivalent modulo \ACI or \ACUI iff $\Var(s) = \Var(t)$ \cite{DBLP:journals/tcs/BaaderB88}. 
For \AI, local finiteness was show by McLean~\cite{DMcLean54} (see also \cite{SiSz82}).
\end{proof}

Together with \autoref{noeth:prop}, the previous two propositions show that the unrestricted unification types of $\AI$, $\ACI$, and $\ACUI$ cannot be zero.
Since we already know that the unrestricted unification types of these theories are at least infinitary 
(see \autoref{EL:D:cor}, \autoref{ACI:cor}, and \autoref{ACUI:cor}),
this yields the following theorem.

\begin{thm}\label{AI:ACI:thm}
	The unrestricted unification type of $\AI$, $\ACI$, and $\ACUI$ for elementary unification is infinitary.
\end{thm}

The theory $\EL$ is regular, but not locally finite since the terms in the sequence
\[
   x,\ \ \exists r.x,\ \ \exists r.(\exists r.x),\ \ \exists r.(\exists r.(\exists r.x)), \ldots 
\]
are pairwise not $\id_\EL$-equivalent. Thus, \autoref{locally:finite:is:noeth:prop} cannot be applied to \EL. However, a proof very similar
to the one of \autoref{locally:finite:is:noeth:prop} can be used to show that $\EL$ is Noetherian.
The reason is that one can actually bound the role depth of the terms relevant in the proof.

The \emph{role depth $\rd(C)$} of an \EL concept (viewed as a term with variables from $V$) is the the maximal nesting of existential restrictions,
i.e., 
\begin{itemize}
	\item
	$\rd(x) := \rd(\top) := 0$ for $x\in V$,
	\item
	$\rd(C\sqcap D) := \max\{\rd(C),\rd(D)\}$ and $\rd(\exists r.C) := 1+\rd(C)$.
\end{itemize}
 
The following lemma is an easy consequence of the fact that $\idEL$ preserves the role depth as well as the set of role names occurring in a concept description, and applying a substitution to a concept description can only increase the role depth and add role names, but not decrease the role depth or remove role names.

\begin{lem}\label{bounds:lem}
	Let $\sigma, \theta$ be substitutions such that $\sigma\leqEL \theta$. 
	Then the following holds for all $x\in V$: the role depth of $\sigma(x)$ is bounded by the role depth of $\theta(x)$, and the role names occurring in $\sigma(x)$ also occur in $\theta(x)$.
\end{lem}

As a consequence of \autoref{regular:theory:prop} and \autoref{bounds:lem} we know that, for a given substitution $\theta$, the set of more general
substitutions is finite up to the equivalence relation $\simEL$. In fact, \autoref{regular:theory:prop} shows that one can restrict the attention
to substitutions $\sigma$ satisfying $\Dom(\sigma)\cup\VRan(\sigma)\subseteq\Dom(\theta)\cup\VRan(\theta)$.
In addition, for all $x\in \Dom(\sigma)$, the concept descriptions $\sigma(x)$ are built using role names occurring in $\theta(x)$ as well
as variables from $\Dom(\theta)\cup\VRan(\theta)$, and have a role depth that is bounded by the one of $\theta(x)$. 
It is well-known~\cite{DBLP:journals/corr/abs-1006-2289} that there are up to $\idEL$-equivalence only finitely many $\EL$ concept descriptions
satisfying these properties.

\begin{lem}\label{finite:EL:more:gen:lem}
	For a given substitution $\theta$, the set $\{\sigma \mid \sigma\leqEL \theta\}$ is finite up to $\simEL$-equivalence.
\end{lem}

This lemma obviously implies that $\EL$ is Noetherian (see the proof of \autoref{locally:finite:is:noeth:prop}). Hence, an application of \autoref{noeth:prop} yields that the unrestricted unification type of $\EL$ cannot be zero. Since we already know that the unrestricted unification type of \EL is at least infinitary (see \autoref{EL:D:cor}), we obtain the following result.

\begin{thm}\label{EL:thm}
	The unrestricted unification type of $\EL$ for elementary unification is infinitary.
\end{thm}

\subsection{Monoidal theories}\label{monoidal:theories:sec}

The theories \ACU and \ACUI are examples of a regular theory that belongs to the class of monoidal theories. For \ACU, which unlike \ACUI is not locally finite, we proved in \cite{BaaderG25} directly that its unrestricted unification type cannot be zero. In this section, we generalize this result to a whole class of regular monoidal theories, which also contains \FLzero and \ACUh. To this purpose, we proceed as follows:
\begin{enumerate}
	\item 
	We recall some relevant aspects of the relation between monoidal theories and semirings.
	\item 
	We define the subclass of \spEName\ monoidal theories, and then show that the regular theories in this subclass are Noetherian. % w.r.t. the unrestricted instantiation preorder, and 
        Thus, they cannot have unrestricted unification type zero.
	\item 
	Finally, we establish that \ACU, \FLzero, and \ACUh are restrictive monoidal theories, and thus, their unrestricted unification type is infinitary.
\end{enumerate}
Let us start by introducing the notion of a semiring. A \emph{semiring} $\sR$ is a tuple $(S,\sRSum,\sZero,\sRPr,\sOne)$, where $\sRSum$ and $\sRPr$ are two binary operations over the set $S$, and $\sZero$ and $\sOne$ are two constant elements, such that:
\begin{itemize}
	\item 
	$(S,\sRSum,\sZero)$ is a commutative monoid.
	\item 
	$(S,\sRPr,\sOne)$ is a monoid.
	\item 
	the operator $\sRPr$ distributes over $\sRSum$ from left and right, i.e., for all $\sV, \sVb, \sVc \in S$:
	\[
	    \sV \sRPr (\sVb \sRSum \sVc) = (\sV \sRPr \sVb) \sRSum (\sV \sRPr \sVc)\ \ \ \mbox{and}\ \ \ (\sV \sRSum \sVb ) \sRPr \sVc = (\sV \sRPr \sVc) \sRSum (\sVb \sRPr \sVc).
	\]
	\item 
	$\sV \sRPr \sZero = \sZero \sRPr \sV = \sZero$ for all $\sV \in S$.
\end{itemize}
The operators $\sRSum$ and $\sRPr$ are usually called the \emph{addition} and \emph{multiplication} of the semiring, respectively. %

Nutt showed in \cite{Nutt90} that solving unification problems in a monoidal theory $E$ can be reduced to solving linear equations over an associated semiring $\sR_E=(S_E,\sRSum,\sZero,\sRPr,\sOne)$. The precise details of how to define $\sR_E$ from $E$ can be found in \cite{Nutt90,BaaderN96}.
Before sketching on the general level how one gets from unification problems modulo a monoidal theory to linear equations over the associated semiring,
let us reconsider \autoref{nuof:acui:ex} under this aspect. 

\begin{exa}
The semiring corresponding to \ACUI is the Boolean semiring
$\BS = (\{0,1\},{\oplus},{\otimes},0,1)$, where $\oplus$ is disjunction and ${\otimes}$ is conjunction~\cite{Nutt90}.
The left-hand side $f(x_1,f(x_2,x_3))$ and  the right-hand side $f(x_4, x_5)$ of the unification problem \eqref{nuof:acui:problem}
can respectively be represented by the row vectors $(1,1,1,0,0)$ and $(0,0,0,1,1)$ over $\BS$. The linear equation corresponding to \eqref{nuof:acui:problem}
is then 
$$
(1,1,1,0,0)\otimes\vec{x} = (0,0,0,1,1)\otimes\vec{x},
$$ 
where $\vec{x}$ is the column vector built from the variables $x_1,\ldots,x_5$.
The columns of the matrix
$$
\left(
                \begin{array}{c@{\ \ \ \ \ }c@{\ \ \ \ \ }c@{\ \ \ \ \ }c@{\ \ \ \ \ }c@{\ \ \ \ \ }c}
                        1 & 1 & 0 &  0 & 0 & 0 \\[.2em]
                        0 & 0 & 1 &  1 & 0 & 0 \\[.2em]
                        0 & 0 & 0 &  0 & 1 & 1 \\[.2em]
                        1 & 0 & 1 &  0 & 1 & 0 \\[.2em]
                        0 & 1 & 0 &  1 & 0 & 1 \\[.2em]
                \end{array}
\right)
$$
constitute a minimal generating set of the solutions of this equation. Note that each such column corresponds to one of the substitutions 
$\theta_1, \theta_2,\ldots, \theta_6$ considered in \autoref{nuof:acui:ex}, and thus to one of the variables $x_1,\ldots,x_6\in\VRan(\theta)$
for the mgu $\theta$ defined in \eqref{nuof:unifier}.
For example, that $\theta(x_4) = f(x_1,f(x_3,x_5))$ can be seen from the fourth row of the matrix, taking into account that the first column
corresponds to $x_1$, the second to $x_2$, etc.\ in the variable range of $\theta$.
\end{exa}

Considered from a more general point of view, a key component of showing the results in \cite{Nutt90} is that, given a finite set of variables $X_n = \{x_1,\ldots,x_n\}$, the $\id_E$-equivalence class of a term $t \in T(\Sigma,X_n)$ can be uniquely represented by a vector $\vec{v}_n(t) = (\occVar{t}{1}, \ldots, \occVar{t}{n})$, where $\occVar{t}{1}, \ldots, \occVar{t}{n}$ are elements of $\sR_E$ (this is a consequence of Corollary 5.11 in \cite{Nutt90}).
Based on this representation, a substitution $\sigma$ with $\Dom(\sigma) \subseteq X_n$ and $\VRan(\sigma) \subseteq \{y_1,\ldots,y_m\}$ ($n,m \geq 1$) can be represented as an $n \times m$ matrix with rows $\vec{v}_m(\sigma(x_1)),\ldots,\vec{v}_m(\sigma(x_n))$. 
Moreover, applying a substitution $\sigma$ to a term $t$ corresponds to multiplying the vector for $t$ with the matrix for $\sigma$,
i.e., the representation $\vec{v}_m(\sigma(t))$ of $\sigma(t)$ can be obtained by computing the following sums over $\sR_E$ (for all $j, 1\leq j \leq m$):
\begin{equation}\label{semiring:sums}
	\occVar{\sigma(t)}{j}\ =\  \sRBSum\limits_{k=1}^n \occVar{t}{k} \sRPr \occVar{\sigma(x_k)}{j}.
\end{equation}
For simplicity, given a substitution $\sigma$, we will from now on write $(\occVar{\sigma}{i1},\ldots,\occVar{\sigma}{im})$ instead of $(\occVar{\sigma(x_i)}{1},\ldots,\occVar{\sigma(x_i)}{m})$ to refer to the vector $\vec{v}_m(\sigma(x_i))$.

\begin{exa}
	The semiring $\sR_\ACU$ associated to \ACU corresponds to the semiring of natural numbers $(\mathbb{N},+,0,\cdot,1)$~\cite{DBLP:journals/jsc/Baader89,Nutt90}. 
The $\id_\ACU$-equivalence class of a term $t \in T(\Sigma,X_n)$ can be uniquely represented by the vector $\vec{v}_n(t) = (\occVar{t}{1}, \ldots, \occVar{t}{n})$, where $\occVar{t}{i}$ is the number of occurrences of the variable $x_i$ in $t$ (Lemma~10.3.1 of \cite{DBLP:books/daglib/0092409}).
	
	For example, if $t = f(x_1,f(x_2,x_2))$ and $\sigma = \{x_1\mapsto x_2, x_2 \mapsto f(x_1,x_1)\}$, then
	$\vec{v}_2(t) = (1,2)$ and the $2\times 2$ matrix representing $\sigma$ has the rows $\vec{v}_2(\sigma(x_1))=(0,1)$ and $\vec{v}_2(\sigma(x_2)) = (2,0)$. 
	The vector representing $\sigma(t) = f(x_2,f(f(x_1,x_1),f(x_1,x_1)))$ is $\vec{v}_2(\sigma(t))=(4,1)$. Note that $4 = 1 \cdot 0 + 2 \cdot 2$ and $1 = 1 \cdot 1 + 2 \cdot 0$. In fact, these two expressions correspond to the instances over $\sR_\ACU$ of the sums defined in \eqref{semiring:sums}, when computing $\vec{v}_2(\sigma(t))$ for the substitution $\sigma$ and term $t$ in our example.
\end{exa}

Let us continue by defining the subclass of \spEName\ monoidal theories. A \emph{total preorder} on a set $S$ is a reflexive and transitive binary relation $\leq$ on $S$ such that any two elements $a, b$ of $S$ are comparable, i.e., $a\leq b$ or $b\leq a$. As usual, we denote the strict part of $\leq$ as $<$, i.e., $a < b$ if
$a\leq b$, but $b\not\leq a$.
\begin{defi}\label{monoidal:special:def}
	Given a monoidal equational theory $E$, we say that $E$ is \emph{\spEName} if there is a \emph{total preorder} $\sPre{\sPreSymbol_E}$ on the carrier set $S_E$ of the semiring $\sR_E$ associated to $E$ such that
	\begin{enumerate}
		\item\label{pre:cond:1} 
		For all $\sV, \sVb, \sVc \in S_E$ with $\sVc \neq \sZero$: $\sV  \sPreL{\sPreSymbol_E} \sVb$ implies that $\sV \sPreL{\sPreSymbol_E} \sVb \sRPr \sVc$.
		\item\label{pre:cond:2}  
		For all $\sV, \sVb, \sVc \in S_E$: $\sV \sPreL{\sPreSymbol_E} \sVb$ implies that $\sV \sPreL{\sPreSymbol_E} \sVb \sRSum \sVc$.
		\item\label{pre:cond:3}  
		For all $\sV \in S_E$: the set $\{\sVb \mid \sVb \sPre{\sPreSymbol_E} \sV\}$ is finite.
	\end{enumerate}
\end{defi}
It is easy to see that the standard total order $\leq_\mathbb{N}$ on $\mathbb{N}$ satisfies these properties, and thus, \ACU is a \spEName\ monoidal theory. We will later show that \ACUI and \FLzero are also \spEName\ theories.
It is worth asking whether there are natural examples of monoidal theories that are not \spEName. The answer to this question is positive, 
as shown by the following example.

\begin{exa}\label{AG:is:not:restrictive:lem}
        The semiring $\sR_\AG$ associated to the theory \AG of Abelian groups corresponds to the semiring of integers $(\mathbb{Z},+,0,\cdot,1)$~\cite{DBLP:journals/jsc/Baader89,Nutt90}.
	To show that the equational theory \AG is not a \spEName\ theory, we must demonstrate that $(\mathbb{Z},+,0,\cdot,1)$ does not admit a total preorder satisfying the conditions required in \autoref{monoidal:special:def}.

	 Suppose there exists a total preorder $\sPre{\mathbb{Z}}$ on $\mathbb{Z}$ satisfying these conditions. %
        Note that this preorder need not be identical to the usual total order $\leq$ on $\mathbb{Z}$.
	Let $\sV \in \mathbb{Z}$. By Condition~\ref{pre:cond:3}, the set $\{\sVc \mid \sVc \sPre{\mathbb{Z}} \sV\}$ is finite. Hence, since $\sPre{\mathbb{Z}}$ is a total preorder and $\mathbb{Z}$ is infinite, there must exist an element $\sVb \in \mathbb{Z}$ such that $\sV \neq \sVb$ and $\sV \sPreL{\mathbb{Z}} \sVb$.
	It then follows by Condition~\ref{pre:cond:2} that
	%
	%\[
        \begin{equation}\label{integer:eqn}
	   \sV \sPreL{\mathbb{Z}} \sVb + \sVc\ \mbox{for all } \sVc \in \mathbb{Z}. % 
        \end{equation}
	%\]
	%
	However, if we set $\sVc := \sV-\sVb$, then $\sV = \sVb + \sVc$. This means that $\sVb + \sVc \sPre{\mathbb{Z}} \sV$ since $\sPre{\mathbb{Z}}$ is reflexive, which contradicts \eqref{integer:eqn}.
	Hence, such a preorder $\sPre{\mathbb{Z}}$ cannot exist, which shows that the monoidal theory \AG is not \spEName.
\end{exa}

We now turn our attention to showing that regular and \spEName\ monoidal theories are Noetherian w.r.t.\ the unrestricted instantiation preorder. To this end, it is convenient to first make some general assumptions for the rest of the section:
\begin{itemize}
	\item
	We assume that $E$ is a monoidal theory that is also \spEName. 
	\item 
	We will use $\sPre{\sPreSymbol_E}$ to refer to a fixed total preorder on $S_E$ satisfying the conditions in \autoref{monoidal:special:def}.
	\item
	Given $n \geq 0$, we use $X_n$ to denote an indexed set of $n$ variables $\{x_1,\ldots,x_n\}$. In addition, we denote by $\allPairs{n}$ the set of pairs $\{1,\ldots,n\}\times\{1,\ldots,n\}$.
\end{itemize}

 Next, we introduce some notions and show a property about $\leq_E^V$ (in \autoref{three:substitutions:lem}), which will be useful later on.
\begin{defi}\label{rel:greater:than:def}
	Let $n \geq 0$ and $P \subseteq \allPairs{n}$. Additionally, let $\sigma_1$ and $\sigma_2$ be substitutions whose domains and variable ranges are contained in $X_n$. We say that \emph{$\sigma_1$ is greater than $\sigma_2$ w.r.t.\ $P$ (denoted as $\sigma_1\grRel{P}\sigma_2$)} if
	\begin{itemize}
		\item 
		for all $(i,j) \in P$ and $k \in \{1,\ldots,n\}$: $\occVar{\sigma_1}{ij}\ \sPreGr{\sPreSymbol_E} \ \occVar{\sigma_2}{ik}$,
	\end{itemize}
	and that \emph{$\sigma_1$ and $\sigma_2$ coincide on $P$} if 
	\begin{itemize}
		\item 
		for all $(i,j) \in P$: $\occVar{\sigma_1}{ij}\ = \occVar{\sigma_2}{ij}$.
	\end{itemize}
\end{defi}

It is not hard to show that $\grRel{P}$ is a transitive relation.

\begin{lem}\label{lem:greater:than:is:transitive}
	Let $n \geq 0$ and $P \subseteq \allPairs{n}$. The relation $\grRel{P}$ is transitive.
\end{lem}
\begin{proof}
	Let $\sigma_1, \sigma_2$ and $\sigma_3$ be substitutions such that $\sigma_1 \grRel{P} \sigma_2$ and $\sigma_2 \grRel{P} \sigma_3$. To prove transitivity of $\grRel{P}$, we need to show that $\sigma_1 \grRel{P} \sigma_3$. 
	
	Let us take any pair $(i,j) \in P$. Since $\sigma_1 \grRel{P} \sigma_2$, we have that $\occVar{\sigma_1}{ij}\ \sPreGr{\sPreSymbol_E} \ \occVar{\sigma_2}{ik}$ for all $k \in \{1,\ldots,n\})$. In particular, this yields $\occVar{\sigma_1}{ij}\ \sPreGr{\sPreSymbol_E} \ \occVar{\sigma_2}{ij}$. Furthermore, $\sigma_2 \grRel{P} \sigma_3$ means that $\occVar{\sigma_2}{ij}\ \sPreGr{\sPreSymbol_E} \ \occVar{\sigma_3}{ik}$ for all $k \in \{1,\ldots,n\}$.
	Consequently, transitivity of $\sPreGr{\sPreSymbol_E}$ implies that
	$
	\occVar{\sigma_1}{ij}\ \sPreGr{\sPreSymbol_E} \ \occVar{\sigma_3}{ik}\ \ \mbox{for all } k \in \{1,\ldots,n\}.
	$
	Thus, since $(i,j)\in  P$ was arbitrarily chosen, we can conclude that $\sigma_1 \grRel{P} \sigma_3$.
\end{proof}

The following property of $\leq_E^V$ is needed to prove \autoref{three:substitutions:lem}.

\begin{lem}\label{dom:range:lambda:Xn:lem}
	Let $\sigma_1$ and $\sigma_2$ be substitutions such that $\sigma_1 \leq_E^V \sigma_2$ and
	\[
	\Dom(\sigma_1) \cup \Dom(\sigma_2) \cup \VRan(\sigma_1) \cup \VRan(\sigma_2) \subseteq X_n.
	\]
	Then, there is a substitution $\lambda$ such that $\lambda\sigma_1 \id_E \sigma_2$ and $\Dom(\lambda) \cup \VRan(\lambda) \subseteq X_n$.
\end{lem}
\begin{proof}
	Since $\sigma_1 \leq_E^V \sigma_2$, there is a substitution $\lambda$ such that $\lambda\sigma_1 \id_E \sigma_2$.
	Let $\VRan(\lambda) \setminus X_n = \{z_1,\ldots,z_k\}$. Consider the substitution $\theta := \{z_1 \mapsto x_1,\ldots,z_k \mapsto x_1\}$, and define the substitution $\lambda'$ as follows:
	\begin{itemize}
		\item
		$\lambda'(x_i) := \theta (\lambda(x_i))$ for all $i \in \{1,\ldots,n\}$, and
		\item
		$\lambda'(y) := y$ for all $y \not \in X_n$. 
	\end{itemize}
	The definition of $\lambda'$ ensures that $\Dom(\lambda') \cup \VRan(\lambda') \subseteq X_n$, and $\Dom(\sigma_1) \cup \Dom(\sigma_2) \subseteq X_n$ implies that $\lambda'\sigma_1(y) = y = \sigma_2(y)$ for all $y \not \in X_n$. 
	Finally, since $\lambda'$  and $\theta\lambda$ coincide on $X_n$, $\id_E$ is closed under substitutions, and $\{z_1,\ldots,z_k\} \cap \VRan(\sigma_2) = \emptyset$, we obtain
	\[
	\lambda'\sigma_1(x_i) = \theta (\lambda \sigma_1(x_i)) \id_E \theta (\sigma_2(x_i)) = \sigma_2(x_i)\ \ \mbox{for all } i \in \{1,\ldots,n\}.
	\]
	Hence, we have shown that $\lambda'\sigma_1 \id_E \sigma_2$. This concludes our proof.
\end{proof}

We are now ready to prove \autoref{three:substitutions:lem}.

\begin{lem}\label{three:substitutions:lem}
	Let $n \geq 0$ and $P \subseteq \allPairs{n}$. In addition, let $\sigma_1$, $\sigma_2, \sigma_3$ be substitutions such that:	
	\begin{itemize}
		\item 
		$\Dom(\sigma_\ell) \cup \VRan(\sigma_\ell) \subseteq X_n$ for $\ell \in \{1,2,3\}$,
		\item 
		$\sigma_1 \leq_E^V \sigma_2$,
		\item 
		$\sigma_1$ is greater than $\sigma_2$ w.r.t.\ $P$, and
		\item 
		$\sigma_1$ and $\sigma_3$ coincide on $\allPairs{n}\setminus P$.
	\end{itemize} 
	Then, $\sigma_3 \leq_E^V \sigma_2$.
\end{lem}
\begin{proof}
	Since $\sigma_1 \leq_E^V \sigma_2$, there is a substitution $\lambda$ such that $\lambda\sigma_1 \id_E \sigma_2$. By \autoref{dom:range:lambda:Xn:lem}, we can assume that $\Dom(\lambda) \cup \VRan(\lambda) \subseteq X_n$.
	Hence, since $\Dom(\sigma_1) \cup \Dom(\sigma_3) \cup \Dom(\lambda) \subseteq X_n$, we know that $\Dom(\lambda\sigma_1) \subseteq X_n$  and $\Dom(\lambda\sigma_3) \subseteq X_n$. To prove that $\sigma_3 \leq_E^V \sigma_2$, it thus suffices to show that
	\begin{equation}\label{to:show}
		\lambda\sigma_1(x_i)\ \id_E\ \lambda\sigma_3(x_i)\ \ \mbox{for all } i \in \{1,\ldots,n\}.
	\end{equation}
	Let $i \in \{1,\ldots,n\}$. Since $\VRan(\sigma_1) \cup \VRan(\lambda) \subseteq X_n$, the $\id_E$-equivalence class of $\lambda\sigma_1(x_i)$ can be uniquely represented with a vector of the form $(\occVar{\lambda\sigma_1}{i1},\ldots, \occVar{\lambda\sigma_1}{in})$, where each value $\occVar{\lambda\sigma_1}{ij}$ ($1 \leq j \leq n$) is determined by the following expression:
	\[
	\occVar{\lambda\sigma_1}{ij}\ =\ \sRBSum_{k=1}^n \occVar{\sigma_1}{ik} \sRPr \occVar{\lambda}{kj}.
	\]
	Since $\lambda\sigma_1(x_i) \id_E \sigma_2(x_i)$, we thus obtain 
	\begin{equation}\label{expre:coeff:instantiation}
		\occVar{\lambda\sigma_1}{ij}\  =\ \sRBSum_{k=1}^n \occVar{\sigma_1}{ik} \sRPr \occVar{\lambda}{kj}\ =\ \occVar{\sigma_2}{ij}\ \ \mbox{for all } j \in \{1,\ldots,n\}.
	\end{equation}
	Consider any pair $(i,k) \in P$. Since $\sigma_1\grRel{P} \sigma_2$, this means that $\occVar{\sigma_1}{ik} \sPreGr{\sPreSymbol_E} \occVar{\sigma_2}{ij}$ for all $j \in \{1,\ldots,n\}$.
	Hence, it must be the case that $\occVar{\lambda}{kj} = \sZero$ for all $j \in \{1,\ldots,n\}$. Otherwise, there would be an index $j \in \{1,\ldots,n\}$ such that
	\begin{itemize}
		\item 
		$\occVar{\sigma_1}{ik} \sRPr \occVar{\lambda}{kj} \sPreGr{\sPreSymbol_E} \occVar{\sigma_2}{ij}$ (by \eqref{pre:cond:1} in \autoref{monoidal:special:def}), and hence, 
		\item 
		the whole sum in \eqref{expre:coeff:instantiation} yields a value $\sV \sPreGr{\sPreSymbol_E} \occVar{\sigma_2}{ij}$ (by \eqref{pre:cond:2} in \autoref{monoidal:special:def}).
	\end{itemize} 
	Thus, since $\sPreGr{\sPreSymbol_E}$ is irreflexive, the expression in \eqref{expre:coeff:instantiation} could not be true for such in index $j$.
	Now, since $\sZero$ is an \emph{annihilator} for $\sRPr$ and an \emph{identity} for $\sRSum$, \eqref{expre:coeff:instantiation} can be turned into:       
	\[
	\occVar{\lambda\sigma_1}{ij}\  =\ \sRBSum_{\substack{k=1\\[.3em] (i,k) \in \allPairs{n}\setminus P}}^n \occVar{\sigma_1}{ik} \sRPr \occVar{\lambda}{kj}\ =\ \occVar{\sigma_2}{ij}\ \ \mbox{for all } j \in \{1,\ldots,n\}.
	\]
	Regarding $\lambda\sigma_3$, we also have $\VRan(\sigma_3) \cup \VRan(\lambda) \subseteq X_n$. Hence, the $\id_E$-equivalence class of $\lambda\sigma_3(x_i)$ has a representation of the form
	$
	(\occVar{\lambda\sigma_3}{i1},\ldots, \occVar{\lambda\sigma_3}{in})
	$,
	where (for $1 \leq j \leq n$):
	\[
	\occVar{\lambda\sigma_3}{ij}\  =\ \sRBSum_{\substack{k=1\\[.3em] (i,k) \in \allPairs{n}\setminus P}}^n \occVar{\sigma_3}{ik} \sRPr \occVar{\lambda}{kj}.
	\]
	But then, since $\sigma_1$ and $\sigma_3$ coincide on $\allPairs{n} \setminus P$, we obtain
	\[
	\occVar{\lambda\sigma_3}{ij}\  =\ \sRBSum_{\substack{k=1\\[.3em] (i,k) \in \allPairs{n}\setminus P}}^n \occVar{\sigma_3}{ik} \sRPr \occVar{\lambda}{kj}\ =\ \sRBSum_{\substack{k=1\\[.3em] (i,k) \in \allPairs{n}\setminus P}}^n \occVar{\sigma_1}{ik} \sRPr \occVar{\lambda}{kj}\ =\ \occVar{\lambda\sigma_1}{ij}.
	\]
	Hence, we have shown that $(\occVar{\lambda\sigma_1}{i1},\ldots, \occVar{\lambda\sigma_1}{in}) = (\occVar{\lambda\sigma_3}{i1},\ldots, \occVar{\lambda\sigma_3}{in})$, i.e., $\lambda\sigma_1(x_i)$ and $\lambda\sigma_3(x_i)$ have the same representation over $S_E$, which yields $\lambda\sigma_1(x_i) \id_E \lambda\sigma_3(x_i)$. 	
	Since $i$ was arbitrarily selected, we have thus proved the claim in \eqref{to:show}. This concludes the proof of the lemma.
\end{proof}

We are now ready to move into the final part of our proof that regular, \spEName\ monoidal theories are Noetherian. The main idea is to show the following: if there is an infinite decreasing chain w.r.t.\ $>_E^V$, then such a chain contains three substitutions $\sigma_1, \sigma_2$ and $\sigma_3$ such that
\begin{itemize}
	\item 
	$\sigma_3 >_E^V \sigma_2 >_E^V \sigma_1$, and
	\item
	$\sigma_1,\sigma_2$ and $\sigma_3$ satisfy the hypothesis of \autoref{three:substitutions:lem}.
\end{itemize}
This implies that $\sigma_2 \sim_E \sigma_3$, which contradicts the existence of such an infinite chain. We will establish this contradiction with the help of the following result.
\begin{lem}\label{contradiction:lem}
	Let $n \geq 0$ and $P \subseteq \allPairs{n}$.
	Suppose there is an infinite decreasing chain $\seq: \sigma_1>_E^V \sigma_2>_E^V \cdots$ such that, for all $\ell \geq 1$:
	\[
	\Dom(\sigma_\ell) \cup \VRan(\sigma_\ell) \subseteq X_n\ \ \ \mbox{and}\ \ \ \sigma_{\ell+1} \grRel{P}\sigma_\ell.
	\]
	Then there exists $(i,j) \in \allPairs{n} \setminus P$ and an infinite decreasing chain $\tau_1>_E^V \tau_2>_E^V \cdots$ such that, for all  $p \geq 1$:
	\[
	\Dom(\tau_p) \cup \VRan(\tau_p) \subseteq X_n\ \ \ \mbox{and}\ \ \ \tau_{p+1} \grRel{P \cup \{(i,j)\}} \tau_p.
	\]
\end{lem}
\begin{proof}
	Suppose such an infinite decreasing chain $\seq$ exists. We claim that there is a pair $(i,j)$ in $\allPairs{n} \setminus P$ such that the sequence of values $\occVar{\sigma_1}{ij}, \occVar{\sigma_2}{ij},\ldots$ is \emph{not bounded} w.r.t.\ $\sPreGr{\sPreSymbol_E}$, i.e.,
	\begin{equation}\label{no:bound}
		\text{for all } \sV \in S_E, \text{ there exists } \ell \geq 1 \text{ such that:}\ \occVar{\sigma_\ell}{ij} \sPreGr{\sPreSymbol_E} \sV.
	\end{equation}
	Suppose this is not the case. Then, given $(i,j) \in \allPairs{n} \setminus P$, there is $\sV \in S_E$ such that $\occVar{\sigma_\ell}{ij} \not \sPreGr{\sPreSymbol_E} \sV$ for all $\ell \geq 1$. Since $\sPre{\sPreSymbol_E}$ is a total preorder, it follows that $\occVar{\sigma_\ell}{ij}  \sPre{\sPreSymbol_E} \sV$ for all $\ell \geq 1$.
	\eqref{pre:cond:3} in \autoref{monoidal:special:def} yields that $\{\occVar{\sigma_\ell}{ij} \mid \ell \geq 1\}$ is a finite set. Hence, there must be infinitely many substitutions in $\seq$ that coincide on $\allPairs{n} \setminus P$. 
	As a consequence, we can select substitutions $\sigma_{\ell_1}, \sigma_{\ell_2}$ and $\sigma_{\ell_3}$ in $\seq$ such that:
	\begin{itemize}
		\item 
		$\sigma_{\ell_3} >_E^V \sigma_{\ell_2} >_E^V \sigma_{\ell_1}$,
		\item 
		$\sigma_{\ell_1}$ is greater than $\sigma_{\ell_2}$ w.r.t.\ $P$, and
		\item
		$\sigma_{\ell_1}$ and $\sigma_{\ell_3}$ coincide on $\allPairs{n}\setminus P$.
	\end{itemize}
	Hence, an application of \autoref{three:substitutions:lem} yields $\sigma_{\ell_3} \leq_E^V \sigma_{\ell_2}$, which contradicts $\sigma_{\ell_3} >_E^V \sigma_{\ell_2}$. Thus, we have shown that \eqref{no:bound} is true.
	
	Based on such a pair $(i,j)$, we define the infinite subsequence $\tau_1, \tau_2, \cdots$ of $\seq$ as follows:
	\begin{enumerate}
		\item 
		$\tau_1 = \sigma_1$.
		\item 
		The set $\{\occVar{\sigma_1}{ik} \mid 1 \leq k \leq n\}$ is finite. Then, since $\sPre{\sPreSymbol_E}$ is a total preorder, there is $\sV \in S_E$ such that: 
		\[
		\occVar{\sigma_1}{ik} \sPre{\sPreSymbol_E} \sV \ \ \mbox{for all } k \in \{1,\ldots,n\}.
		\]
		Hence, by \eqref{no:bound} and transitivity of $\sPre{\sPreSymbol_E}$,¸ there is $p > 1$ such that $\occVar{\sigma_p}{ij} \sPreGr{\sPreSymbol_E} \occVar{\sigma_1}{ik}$ for all $k \in \{1,\ldots,n\}$.
		We choose $\tau_2$ as $\sigma_p$. The following arguments show that $\tau_1$ and $\tau_2$ satisfy the properties required of $\tau_1,\tau_2, \cdots$ w.r.t.\ $P \cup \{(i,j)\}$:
		\begin{itemize}
			\item 
			Since $\sigma_1 >_E^V \sigma_p$, we have that $\tau_1 >_E^V \tau_2$.
			\item 
			By selection of $p$, we know that $\sigma_p \grRel{\{(i,j)\}}\sigma_1$. Moreover, since $\sigma_{\ell+1} \grRel{P}\sigma_\ell$ for all $\ell \geq 1$, transitivity of $\grRel{P}$ yields that $\sigma_p \grRel{P} \sigma_1$. Consequently, we can conclude that $\tau_2$ is greater than $\tau_1$ w.r.t.\ $P \cup \{(i,j)\}$.
		\end{itemize}
		\item 
		Once we fix $\sigma_p$, the same arguments given in (2) can be applied to $\sigma_p$ to obtain $q > p$ such that $\occVar{\sigma_{q}}{ij}  \sPreGr{\sPreSymbol_E} \occVar{\sigma_{p}}{ik}$ for all $k \in \{1,\ldots,n\}$.
		We select $\tau_3$ as $\sigma_q$. The same arguments yield that $\tau_2$ and $\tau_3$ are as required.
		\item 
		By repeating (\emph{ad infinitum}) the described selection process, we can extract from $\seq$ an appropriate remaining sequence of substitutions $\tau_4\tau_5\cdots$.
	\end{enumerate}
	The described process ensures that each selected substitution $\tau_\ell$ ($\ell \geq 1$) belongs to $\seq$, and hence, $\Dom(\tau_\ell) \cup \VRan(\tau_\ell) \subseteq X_n$. Thus, $\tau_1\tau_2\cdots$ is an infinite decreasing chain satisfying the claim of the lemma.
\end{proof}

Finally, by using the previous lemma, we can show the main result of this section.

\begin{prop}\label{regular:spe:monoidal:implies:noetherian:lem}
	Let $E$ be a monoidal theory. If $E$ is regular and \spEName, then $E$ is Noetherian. %
\end{prop}
\begin{proof}
	Let $E$ be a regular, \spEName\ monoidal theory. Suppose that $E$ is not Noetherian, i.e., there exists a strictly decreasing chain $\theta_1 >_E^V \theta_2 >_E^V \theta_3 >_E^V\cdots$ of substitutions. Transitivity of $\leq_E^V$ yields that $\theta_i \leq_E^V \theta_1$ for all $i \geq 1$. Since $E$ is a regular theory, \autoref{regular:theory:prop} thus yields substitutions $\sigma_i$ ($i \geq 1$) such that
	\[
	\theta_i \sim_E^V \sigma_i\ \ \mbox{and}\ \ \Dom(\sigma_i) \cup \VRan(\sigma_i) \subseteq \Dom(\theta_1) \cup \VRan(\theta_1).
	\]
	From this, we obtain the following:
	\begin{itemize}
		\item
		$\seq : \sigma_1 >_E^V \sigma_2 >_E^V \sigma_3 >_E^V\cdots$ is an infinite decreasing chain.
		\item
		$\Dom(\sigma_i) \cup \VRan(\sigma_i) \subseteq X_n$ for all $i \geq 1$, where $x_1,\ldots,x_n$ is an enumeration of $\Dom(\theta_1) \cup \VRan(\theta_1)$.
	\end{itemize}
	Now, let $P_0 = \emptyset$. By \autoref{rel:greater:than:def}, we have %
	%\[
        $
	\sigma_i \grRel{P_0} \sigma_j\ \mbox{for all } i,j \geq 1.
        $
	%\] 
	%
	This means that $\seq$ and $P_0$ satisfy the hypothesis of \autoref{contradiction:lem}. Hence, there is $(i,j) \in \allPairs{n}$ such that $\seq$ contains an infinite decreasing subchain $\tau_1>_E^V \tau_2>_E^V \cdots$ such that the following holds for all $p \geq 1$:
	\[
	\Dom(\tau_p) \cup \VRan(\tau_p) \subseteq X_n\ \ \ \mbox{and}\ \ \ \tau_{p+1} \grRel{\{(i,j)\}} \tau_p.
	\]
	This new chain now satisfies the hypothesis of \autoref{contradiction:lem} w.r.t.\ $P=\{(i,j)\}$. 
	Thus, by a sequence of $n^2 - 1$ further applications of \autoref{contradiction:lem}, we conclude that there is an infinite decreasing chain $\eta_1 >_E^V \eta_2 >_E^V \cdots$ such that the following holds for all $p \geq 1$:
	\[
	\Dom(\eta_p) \cup \VRan(\eta_p) \subseteq X_n\ \ \ \mbox{and}\ \ \ \eta_{p+1} \grRel{\allPairs{n}} \eta_p.
	\]
	Finally, it is not hard to verify that this last chain must contain three substitutions $\eta_{\ell_1}, \eta_{\ell_2}$ and $\eta_{\ell_3}$ such that
	\begin{itemize}
		\item 
		$\eta_{\ell_3} >_E^V \eta_{\ell_2} >_E^V \eta_{\ell_1}$,
		\item 
		$\eta_{\ell_1}$ is greater than $\eta_{\ell_2}$ w.r.t.\ $P=\allPairs{n}$, and
		\item 
		$\eta_{\ell_1}$ and $\eta_{\ell_3}$ coincide on $\allPairs{n}\setminus P = \emptyset$.
	\end{itemize} 
	Hence, an application of \autoref{three:substitutions:lem} yields that $\eta_{\ell_3} \leq_E^V \eta_{\ell_2}$, which contradicts the fact that $\eta_{\ell_3} >_E^V \eta_{\ell_2}$.
	Thus, we have derived a contradiction from our initial assumption that $E$ is not Noetherian. This concludes the proof of the lemma.
\end{proof}

Since the theories \ACU, \FLzero, and \ACUh are obviously regular and it is known that they are monoidal~\cite{DBLP:journals/jsc/Baader89,DBLP:conf/ctcs/Baader89}, 
it remains to show that they are also restrictive.

\begin{prop}
	The equational theories \ACU, \ACUh, and \FLzero are \spEName\ monoidal theories.
\end{prop}

\begin{proof}
We have already discussed why \ACU is a \spEName\ monoidal theory. 

Next, we consider the theory \FLzero, which extends \ACUI with homomorphisms $h_1,\ldots,h_k$ for some $k \geq 1$. Recall that \FLzero and \ACUIh represent the same equational theory.  As pointed out in \cite{DBLP:journals/jsc/BaaderN01}, the semiring $\sR_\ACUIh$ associated to \ACUIh is isomorphic to the semiring $\fSets = (W,\cup,\emptyset,\conc, \{\eword\})$, where:
\begin{itemize}
	\item 
	the elements in $W$ are finite sets of words over the alphabet $\alphb= \{h_1,\ldots,h_k\}$,
	\item 
	the addition operation $\cup$ is union of sets with the empty set $\emptyset$ as unit, and
	\item 
	the multiplication operation $\cdot$ is element-wise concatenation of sets of words with the set $\{\varepsilon\}$ consisting of the empty word $\varepsilon$ as unit.
\end{itemize}
To show that \FLzero is \spEName, we define a relation ${\sPre{\fSets}} \subseteq W \times W$ satisfying the conditions required in \autoref{monoidal:special:def} w.r.t.\ the semiring $\fSets$.
Given $M\in W$, we denote by $m(M)$ the maximal length of a word in $M$. Based on this, we define ${\sPre{\fSets}}$ as follows:
\[
L \sPre{\fSets} M\ \ \mbox{iff}\ \ m(L) \leq m(M) \ \ \mbox{for all } L,M \in W.
\]
It is not hard to verify that $\sPre{\fSets}$ is a total preorder on $W$. It remains to show that $\sPre{\fSets}$ satisfies the three conditions in  \autoref{monoidal:special:def}.
\begin{itemize}
	\item 
	(Condition~\ref{pre:cond:1}) Let $L_1, L_2, L_3 \in W$ such that $L_3 \neq \emptyset$ and $L_1 \sPreL{\fSets} L_2$. We need to show that $L_1 \sPreL{\fSets} L_2 \conc L_3$ (recall that concatenation $\conc$ is the multiplication of $\fSets$).
	From $L_1 \sPreL{\fSets} L_2$, we obtain that $L_1 \sPre{\fSets} L_2$ and $L_2 \not \sPre{\fSets} L_1$. By definition of $\sPre{\fSets}$ this means that $m(L_1) < m(L_2)$.
	Now, since $L_3 \neq \emptyset$, we obtain
	%\[
        $
	m(L_1) < m(L_2) \leq m(L_2 \conc L_3).
        $
	%\]
	%
	Thus, we have shown that $L_1 \sPreL{\fSets} L_2 \conc L_3$.
	\item 
	(Condition~\ref{pre:cond:2}) Let $L_1, L_2, L_3 \in W$ such that $L_1 \sPreL{\fSets} L_2$. We need to show that $L_1 \sPreL{\fSets} L_2 \cup L_3$. As before, $L_1 \sPreL{\fSets} L_2$ yields $m(L_1) < m(L_2)$, and obviously $m(L_2) \leq m(L_2 \cup L_3)$ holds. Hence, $m(L_1) < m(L_2 \cup L_3)$, and thus $L_1 \sPreL{\fSets} L_2 \cup L_3$.
	\item 
	(Condition~\ref{pre:cond:3}) Given $M \in W$, we need to show that $\{L \mid L\sPre{\fSets} M\}$ is a finite set. Every such language $L$ satisfies $m(L) \leq m(M)$. Since $m(M) \in \mathbb{N}$ and $\alphb$ is a finite alphabet, there are only finitely many words over $\alphb$ of length at most $m(M)$.
	Thus, the set $\{L \mid L\sPre{\fSets} M\}$  must be finite.
\end{itemize}
This finishes our proof that the theory \FLzero is restrictive.

For the theory \ACUh, it was observed in \cite{Nutt90} that the semiring $\sR_\ACUh$ associated to \ACUh is isomorphic to the polynomial semiring $(\NPol,\addPol,0,\multPol,1)$, where:
\begin{itemize}
        \item 
        $X=\{x_1,\ldots,x_k\}$ is a finite set of indeterminates,
        \item 
        the elements in $\NPol$ are polynomials defined over the indeterminates in $X$ with natural numbers as coefficients,
        \item
        the addition operation $\addPol$ is the standard sum of polynomials in $\NPol$ with the polynomial $0$ as unit, and
        \item
        the multiplication operation $\multPol$ is the standard multiplication of polynomials in $\NPol$ with the polynomial $1$ as unit.
\end{itemize}
To show that \ACUh is \spEName, we define a relation ${\sPre{\semNPol}} \subseteq \NPol \times \NPol$ satisfying the conditions required in \autoref{monoidal:special:def} w.r.t.\ the semiring $\semNPol$.
Given $p \in \NPol$, we denote by $c(p)$ the maximal coefficient occurring in $p$ and by $e(p)$ the maximal exponent of a variable in $p$. Based on this, we define ${\sPre{\semNPol}}$ as follows:
\[
p \sPre{\semNPol} q\ \ \mbox{iff}\ \ c(p) + e(p) \leq c(q) + e(q) \ \ \mbox{for all } p,q \in \NPol.
\]      
It is not hard to verify that $\sPre{\semNPol}$ is a total preorder on $\NPol$. It remains to show that $\sPre{\semNPol}$ satisfies the three conditions in  \autoref{monoidal:special:def}.
\begin{itemize}
        \item
        (Condition~\ref{pre:cond:1}) Let $p_1, p_2, p_3 \in \NPol$ such that $p_3 \neq 0$ and $p_1 \sPreL{\semNPol} p_2$. We need to show that $p_1 \sPreL{\semNPol} p_2 \multPol p_3$.
        From $p_1 \sPreL{\semNPol} p_2$, we obtain that $p_1 \sPre{\semNPol} p_2$ and $p_2 \not \sPre{\semNPol} p_1$. By definition of $\sPre{\semNPol}$ this means that $c(p_1) + e(p_1) < c(p_2) + e(p_2)$.
        Now, since $p_3 \neq 0$ and the coefficients and variable exponents occurring in $p_2$ and $p_3$ are natural numbers, we obtain that $c(p_2) \leq c(p_2 \multPol p_3)$ and $e(p_2) \leq e(p_2 \multPol p_3)$. This means that
        $
        c(p_1) + e(p_1) < c(p_2) + e(p_2) \leq c(p_2\multPol p_3) + e(p_2 \multPol p_3).
        $
        %\]
        %
        Thus, we have shown that $p_1 \sPreL{\semNPol} p_2 \multPol p_3$.
        \item
        (Condition~\ref{pre:cond:2}) Let $p_1, p_2, p_3 \in \NPol$ such that $p_1 \sPreL{\semNPol} p_2$. We need to show that $p_1 \sPreL{\semNPol} p_2 \addPol p_3$. As before, $p_1 \sPreL{\semNPol} p_2$ yields $c(p_1) + e(p_1) < c(p_2) + e(p_2)$, and obviously $c(p_2) + e(p_2) \leq c(p_2 \addPol p_3) + e(p_2 \addPol p_3)$ holds. Hence, $c(p_1) + e(p_1) < c(p_2 \addPol p_3) + e(p_2 \addPol p_3)$, and thus $p_1 \sPreL{\semNPol} p_2 \addPol p_3$.
        \item
        (Condition~\ref{pre:cond:3}) Given $p \in \NPol$, we need to show that $\{q \mid q\sPre{\semNPol} p\}$ is a finite set. Every such polynomial $q$ satisfies $c(q) + e(q) \leq c(p) + e(p)$. Since $c(p) + e(p) \in \mathbb{N}$ and $X$ is a finite set of indeterminates, there are only finitely many polynomials over $X$ with maximal coefficient at most $c(p)$ and maximal variable exponent at most $e(p)$.
        Thus, the set $\{q \mid q\sPre{\semNPol} p\}$  must be finite.
\end{itemize}
This finishes our proof that the theory \ACUh is also restrictive.
\end{proof}

Hence, since \ACU, \FLzero, and \ACUh are also regular theories, applying \autoref{regular:spe:monoidal:implies:noetherian:lem} yields that they cannot have unrestricted unification type zero. In Section~\ref{lbounds:sect} we have shown that the unification types of these theories are at least infinitary w.r.t.\ the unrestricted instantiation preorder (\autoref{ACU:cor} and \autoref{EL:D:cor}).
Thus, we can conclude that the unrestricted unification type of these three theories is infinitary.

\begin{thm}\label{restrictive:monoidal:thm}
	The unrestricted unification type of \ACU, \FLzero, and \ACUh for elementary unification is infinitary.
\end{thm}